\documentclass[journal,onecolumn,12pt]{IEEEtran}
\usepackage{pslatex}
\usepackage{amsfonts,color,morefloats}
\usepackage{amssymb,amsmath,latexsym,amsthm}
\usepackage{setspace}

\newcommand{\lcm}{{\rm lcm}}

\newcommand{\ord}{{\mathrm{ord}}}
\newcommand{\Z}{\mathbb{{Z}}}

\newcommand{\gf}{{\mathrm{GF}}}

\newcommand{\C}{{\mathcal{C}}}

\newtheorem{theorem}{Theorem}
\newtheorem{lemma}[theorem]{Lemma}
\newtheorem{proposition}[theorem]{Proposition}

\newtheorem{problem}{Open Problem}
\newtheorem{definition}{Definition}
\newtheorem{example}{Example}

\newcommand{\Fp}{\gf(p)}
\newcommand{\Fq}{\gf(q)}
\newcommand{\Fqm}{\gf(q^m)}
\newcommand{\Ft}{\gf(3)}
\newcommand{\Ftm}{\gf(3^m)}
\newcommand{\Fthm}{\gf(3^{\frac{m}{2}})}
\newcommand{\al}{\alpha}
\newcommand{\be}{\beta}
\newcommand{\de}{\delta}
\newcommand{\sig}{\sigma}
\newcommand{\ep}{\epsilon}

\newcommand{\ol}{\overline}
\newcommand{\Tr}{\mathrm{Tr}}
\newcommand{\Trtmt}{\text{Tr}^{3^m}_3}
\newcommand{\ra}[1]{\renewcommand{\arraystretch}{#1}}
\newcommand{\ul}{\underline}
\newcommand{\z}{\zeta_3}

\newcommand{\hm}{\frac{m}{2}}
\newcommand{\mpo}{\frac{m+1}{2}}
\newcommand{\mmo}{\frac{m-1}{2}}
\newcommand{\mpt}{\frac{m+3}{2}}
\newcommand{\mmt}{\frac{m-3}{2}}
\newcommand{\sq}{\sqrt{-1}}

\begin{document}

\title{Narrow-Sense BCH Codes over $\gf(q)$ with Length $n=\frac{q^m-1}{q-1}$\thanks{
The research of C. Ding was supported by the Hong Kong Research Grants Council, under Grant No. 16300415.
The research of M. Xiong was supported by RGC grant number 609513 from Hong Kong. The research of G. Ge was supported by the National Natural Science Foundation of China under Grant Nos. 11431003 and 61571310.}}

\author{Shuxing~Li, Cunsheng~Ding, Maosheng Xiong, and~Gennian Ge   
\thanks{S. Li is with the School of Mathematical Sciences, Zhejiang University, Hangzhou 310027, Zhejiang, China (e-mail: sxli@zju.edu.cn).}
\thanks{C. Ding is with the Department of Computer Science and Engineering, The Hong Kong University of Science and Technology,
Clear Water Bay, Kowloon, Hong Kong, China (e-mail: cding@ust.hk).}
\thanks{M. Xiong is with the Department of Mathematics, The Hong Kong University of Science and Technology, Clear Water Bay, Kowloon, Hong Kong (e-mail: mamsxiong@ust.hk).}
\thanks{G. Ge is with the School of Mathematical Sciences, Capital Normal University, Beijing 100048, China (e-mail: gnge@zju.edu.cn). He is also with Beijing Center for Mathematics and Information Interdisciplinary Sciences, Beijing, 100048, China.}
}

\maketitle

\begin{abstract}
Cyclic codes are widely employed in communication systems, storage devices and consumer electronics, as they have efficient encoding and decoding algorithms. BCH codes, as a special subclass of cyclic codes, are in most cases among the best cyclic codes. A subclass of good BCH codes are the narrow-sense BCH codes
over $\gf(q)$ with length $n=(q^m-1)/(q-1)$. Little is known about this class of BCH codes when $q>2$. The objective of this paper is to study some of the codes within this class. In particular, the dimension, the minimum distance, and the weight distribution of some ternary BCH codes with length $n=(3^m-1)/2$ are
determined in this paper. A class of ternary BCH codes meeting the Griesmer bound is identified. An application of some of the BCH codes in secret sharing is also investigated.
\end{abstract}

\begin{IEEEkeywords}
BCH codes, Bose distance, cyclic codes, minimum distance, quadratic forms, secret sharing, weight distribution.
\end{IEEEkeywords}

\section{Introduction}\label{sec-intro}

Throughout this paper, let $q$ be a power of a prime $p$. An $[n,k,d]$ linear code $\C$ over $\gf(q)$ is a $k$-dimensional subspace of $\gf(q)^n$ with minimum Hamming distance $d$. If in addition of being a linear code, it satisfies the condition that any $(c_0,c_1, \cdots, c_{n-1}) \in \C$ implies $(c_{n-1}, c_0, c_1, \cdots, c_{n-2}) \in \C$, then $\C$ is called a {\em cyclic} code.

By identifying each vector $(c_0,c_1, \cdots, c_{n-1}) \in \gf(q)^n$ with a polynomial
$$
c_0+c_1x+c_2x^2+ \cdots + c_{n-1}x^{n-1} \in R:=\gf(q)[x]/(x^n-1),
$$
a linear code $\C$ of length $n$ over $\gf(q)$ corresponds to a $\gf(q)$-submodule of the ring $R$. Moreover, $\C$ is cyclic if and only if the corresponding submodule is an ideal of $R$.

Note that every ideal of $R$ is principal. Let $\C \subset R$ be a cyclic code of length $n$ over $\gf(q)$, then $\C=\langle g(x) \rangle$ where $g(x) \in \gf(q)[x]$ may be chosen to be monic and to have the smallest degree among all the generators of $\C$. This $g(x)$ is unique and  satisfies $g(x)|(x^n-1)$ and is called the {\em generator polynomial}, and $h(x):=(x^n-1)/g(x)$ is called the {\em parity-check} polynomial of $\C$. If the generator polynomial $g(x)$ (resp. the parity-check polynomial $h(x)$) can be factored into a product of $s$ irreducible polynomials over $\gf(q)$, then $\C$ is called a cyclic code with $s$ zeroes (resp. $s$ nonzeroes). In this paper, we consider only cyclic codes of length $n$ over $\gf(q)$ with $\gcd(n,q)=1$, which implies that the generator polynomial of the code does not have repeated roots.

Let $n$ be a positive integer. Let $m=\ord_n(q)$, that is, $m$ is the smallest positive integer such that $n | q^m-1$. Let $\alpha$ be a generator of $\gf(q^m)^*:=\gf(q^m)\setminus \{0\}$, and put $\beta=\alpha^{(q^m-1)/n}$. Then $\beta$ is a primitive $n$-th root of unity in $\gf(q^m)$. For each $i$ where $0 \leq i \leq n-1$, let $m_i(x)$ denote the minimal polynomial of $\beta^i$ over $\gf(q)$. For each $2 \leq \delta \le n$, define
$$
g_{(n,q,m,\delta)}(x)=\lcm(m_{1}(x), m_{2}(x), \cdots, m_{\delta-1}(x)),
$$
where $\lcm$ denotes the least common multiple of the polynomials. We also define
$$
\tilde{g}_{(n,q,m,\delta)}(x)=(x-1)g_{(n,q,m,\delta)}(x).
$$
Let $\C_{(n,q, m, \delta)}$ and $\tilde{\C}_{(n,q, m, \delta)}$ denote the cyclic codes of length $n$ with generator
polynomials $g_{(n,q, m,\delta)}(x)$ and $\tilde{g}_{(n,q,m,\delta)}(x)$, respectively. Then $\C_{(n,q, m, \delta)}$ is called a \emph{narrow-sense BCH code} with \emph{designed distance} $\delta$, and $\tilde{\C}_{(n,q, m, \delta)}$ is the even-like subcode of $\C_{(n,q, m, \delta)}$. Clearly, we have
$$
\dim(\tilde{\C}_{(n,q, m, \delta)})=\dim(\C_{(n,q, m, \delta)})-1.
$$
Since $g_{(n,q,m,\delta)}(x)$ has $\delta-1$ consecutive roots $\beta^i$ for all $1 \leq i \leq \delta -1$, and $\tilde{g}_{(n,q,m,\delta)}(x)$ has $\delta$ consecutive roots $\beta^i$ for all $0 \leq i \leq \delta -1$, it follows from the BCH bound that the minimum distances of $\C_{(n,q, m, \delta)}$ and $\tilde{\C}_{(n,q, m, \delta)}$ are at least $\delta$ and $\delta +1$, respectively. Due to this fact, $\delta$ is called the designed distance of the code $\C_{(n,q, m, \delta)}$.

It is well known that for two different design distances $\delta$ and $\delta'$, the codes $\C_{(n,q, m, \delta)}$ and $\C_{(n,q, m, \delta')}$ may be the same. The largest designed distance of
$\C_{(n,q, m, \delta)}$ is called the \emph{Bose distance} \cite[p. 205]{MS77}
and is denoted by $d_B$. Since the Bose distance of $\C_{(n,q, m, \delta)}$ serves as the largest lower bound on the minimum distance of the code among all designed distances, it is useful to determine the Bose distance of the BCH code, if the minimum distance cannot be obtained.

The cyclic codes $\C_{(n,q, m, \delta)}$ are treated in almost every book on coding theory. When $n=q^m-1$, the codes
$\C_{(n,q, m, \delta)}$ are called \emph{narrow-sense primitive BCH codes} and have been extensively studied in the literature \cite{AKS,ACS92,AS94,Berlekamp,CC,Charp90,Charp98,YH96,DDZ15,KL72,KLP,KN,Man,Mann,Peter,YF,YZ}. The reader is referred to \cite{DDZ15} for a recent summary of various results on narrow-sense primitive BCH codes. When $n=\frac{q^m-1}{q-1}$, the codes $\C_{(n, q, m, \delta)}$ and $\tilde{\C}_{(n,q, m, \delta)}$ are also interesting, however, very little is known about them when $q>2$.

The objective of this paper is to study some special classes of the narrow-sense BCH codes of length $\frac{q^m-1}{q-1}$ for $q >2$. We employ various tools, including cyclotomic cosets, locator polynomials, nondecreasing sequence decompositions, exponential sums and the theory of quadratic forms over finite fields to study important parameters such as dimensions, Bose distances, minimum distances and weight distributions of these codes. A class of ternary BCH codes meeting the Griesmer bound is presented. Moreover, an application of some of the BCH codes presented in this paper is also investigated.

As will be shown, some narrow-sense BCH codes of length $\frac{q^m-1}{q-1}$ have optimal parameters. To investigate the optimality of some of the codes studied in this paper, we compare them with the tables of the best known linear codes maintained by Markus Grassl at http://www.codetables.de, which are called the {\em Database} later in this paper. In some cases, we also use the tables of the best cyclic codes given in the monograph \cite{Dingbk15} as benchmarks.

\section{Preliminaries}\label{sec-pre}

In this section, we present some necessary background concerning cyclotomic cosets, coset leaders, nondecreasing sequence  decompositions, quadratic forms over finite fields and exponential sums. Some known results about BCH codes are also recalled.

\subsection{Cyclotomic cosets}

Let $n$ be a positive integer such that $\gcd(n,q)=1$. Let $\Z_n$ denote the ring of integers modulo $n$. Let $s$ be an integer with $0 \leq s <n$. The \emph{$q$-cyclotomic coset of $s$ modulo $n$\index{$q$-cyclotomic coset modulo $n$}} is defined by
$$
C_s=\{s, sq, sq^2, \cdots, sq^{\ell_s-1}\} \bmod n \subseteq \Z_n,
$$
where $\ell_s$ is the smallest positive integer such that $q^{\ell_s}s \equiv s \pmod{n}$. $\ell_s$ is also the size of the $q$-cyclotomic coset $C_s$. The smallest nonnegative integer in $C_s$ is called the \emph{coset leader\index{coset leader}} of $C_s$.
Let $\Gamma_{(n,q)}$ be the set of all the coset leaders. Then we have $C_s \cap C_t = \emptyset$ for any distinct elements $s$ and $t$ in  $\Gamma_{(n,q)}$, and
\begin{eqnarray}\label{eqn-cosetPP}
\bigcup_{s \in  \Gamma_{(n,q)} } C_s = \Z_n.
\end{eqnarray}
Namely, the $q$-cyclotomic cosets modulo $n$ form a partition of $\Z_n$.

Suppose $m=\ord_{n}(q)$ and $\beta$ is a primitive $n$-th root of unity in $\gf(q^m)$. Then, the minimal
polynomial $m_{s}(x)$ of $\beta^s$ over $\gf(q)$ is the monic polynomial of the smallest degree over
$\gf(q)$ with $\beta^s$ as a zero. By Galois theory, this polynomial is irreducible over $\gf(q)$ and is given by
\begin{eqnarray*}
m_{s}(x)=\prod_{i \in C_s} (x-\beta^i) \in \gf(q)[x].
\end{eqnarray*}
It follows from (\ref{eqn-cosetPP}) that
\begin{eqnarray*}
x^n-1=\prod_{s \in  \Gamma_{(n,q)}} m_{s}(x),
\end{eqnarray*}
which is the factorization of $x^n-1$ into irreducible factors over $\gf(q)$. Thus, for a cyclic code
$$
\C=\langle g(x) \rangle \subset \gf(q)[x]/(x^n-1),
$$
the generator polynomial $g(x)$ is a product of some $m_{s}(x)$'s. Hence, the degree of $g(x)$ and the dimension of $\C$ can be determined by the size of cyclotomic cosets associated with $g(x)$.

Regarding the size of $q$-cyclotomic cosets modulo $n$, we have the following lemma.

\begin{lemma}{\rm \cite[Theorem 4.1.4]{HP03}}
The size $\ell_s$ of each $q$-cyclotomic coset $\C_s$ is a divisor of $\ord_{n}(q)$, which is the size $\ell_1$ of $C_1$.
\end{lemma}

The following lemma says when $s$ is small, the size of the $q$-cyclotomic coset $\C_s \subset \Z_n$ is always equal to $\ord_n(q)$.

\begin{lemma}{\rm \cite[Lemma 8]{AKS}}
Suppose $q^{\lfloor m/2 \rfloor}<n \leq q^m-1$, where
$m=\ord_n(q)$. Then the $q$-cyclotomic coset $C_s$ has cardinality
$m$ for all $s$ in the range $1 \leq s \leq n q^{\lceil m/2 \rceil}/(q^m-1)$.
\end{lemma}

As a direct consequence, the dimension of some BCH codes can be easily obtained.

\begin{theorem}{\rm \cite[Theorem 10]{AKS}}\label{thm-dim}
Suppose $q^{\lfloor m/2 \rfloor}<n \leq q^m-1$, where
$m=\ord_n(q)$. Then the narrow-sense BCH code $\C_{(n,q,m,\delta)}$ with $\delta$ in the range $2 \leq \delta \leq \min\{\lfloor n q^{\lceil m/2 \rceil}/(q^m-1) \rfloor, n\}$ has dimension
$$
k=n-m\lceil (\delta -1)(1-1/q) \rceil.
$$
\end{theorem}

The following proposition illustrates the close relation between coset leaders and the Bose distance of the narrow-sense BCH code $\C_{(n,q,m,\delta)}$.

\begin{proposition}\label{prop-Bose}
The Bose distance $d_B$ of the code $\C_{(n,q,m,\delta)}$ is a coset leader of a $q$-cyclotomic coset modulo $n$. Moreover, if $\de$ is a coset leader, then $d_B=\de$.
\end{proposition}
\begin{proof}
Suppose $d_B$ is the Bose distance of $\C_{(n,q,m,\delta)}$, then by definition, $\beta^i$ is a root of $g_{(n,q,m,\delta)}(x)$ for each $1 \le i \le d_B-1$. Assume that $d_B$ is not a coset leader. Then there exists some $1 \le j \le d_B-1$, such that $d_B \in C_j$. Therefore, $\beta^{d_B}$ is also a root of $g_{(n,q,m,\delta)}(x)$. This implies $\C_{(n,q,m,\delta)}=\C_{(n,q,m,d_B+1)}$ and $\C_{(n,q,m,\delta)}$ has a designed distance $d_B+1>d_B$. This contradicts to the definition of the Bose distance.

Moreover, if $\de$ is a coset leader, then,
\[\de \not \in \bigcup_{1 \le j \le \de -1} C_j.\]
By the definition of $g_{(n,q,m,\delta)}(x)$, $\be^{\de}$ is not a root of $g_{(n,q,m,\delta)}(x)$. Hence, $\de$ is the largest designed distance of $\C_{(n,q,m,\delta)}$ and we have $d_B=\de$.
\end{proof}

\subsection{Nondecreasing sequence decompositions and coset leaders}

There have been some interesting results on coset leaders of $q$-cyclotomic cosets when $n=q^m-1$ (\cite{DDZ15,Man,YF}). Particularly, in \cite{YF}, the concept of nondecreasing sequence decompositions was proposed and its close relation with coset leaders modulo $q^m-1$ was discussed. In this subsection, we show that this concept is also helpful in determining some special $q$-cyclotomic coset leaders modulo $\frac{q^m-1}{q-1}$. Here we always assume that $q >2$ is a prime power.

Suppose $\ul{v}=(v_{n-1},v_{n-2},\ldots,v_0)$ is a sequence of length $n$ with each component $v_i$ satisfying $0 \le v_i \le q-1$. The sequence $\ul{v}$ is called a {\rm nondecreasing sequence} (NDS) if $v_{i+1} \le v_{i}$ for $0 \le i \le n-2$. Every sequence has a unique {\rm nondecreasing sequence decomposition} as a concatenation $\ul{V_1}\ul{V_2}\ldots\ul{V_r}$ where the $\ul{V_i}$'s are NDSs and $r$ is minimal. Let $\ul{V}=(v_{l},\ldots,v_0)$ and $\ul{W}=(u_{k},\ldots,u_0)$ be two NDSs. We say $\ul{V}=\ul{W}$ if $l=k$ and $v_{l-i}=w_{l-i}$ for $0 \le i \le l$. We say $\ul{V}>\ul{W}$, if either $l>k$ and $v_{l-i}=w_{k-i}$ for $0 \le i \le k$ or there is an integer $0 \le j \le \min\{l,k\}$ such that $v_{l-j} > w_{k-j}$ and $v_{l-i}=w_{k-i}$ for $0 \le i \le j-1$. For two sequences of the same length with the NDS decomposition given by $\ul{v}=\ul{V_1}\ul{V_2}\ldots\ul{V_r}$ and $\ul{w}=\ul{W_1}\ul{W_2}\ldots\ul{W_s}$, we say $\ul{v}>\ul{w}$ if there is an integer $1 \le j \le \min\{r,s\}$ such that $\ul{V_{j}} > \ul{W_{j}}$ and $\ul{V_{i}}=\ul{W_{i}}$ for $1 \le i \le j-1$. We remark that the notion ``$\ul{v}>\ul{w}$'' described above is consistent with the natural inequality $\sum_i v_iq^i>\sum_i w_i q^i$ of the corresponding integers.

Given a positive integer $s$, we may assume that $0<s<n$. Suppose the unique $q$-adic expansion of $s$ is $\sum_{i=0}^{m-1}s_iq^i$,
where $0 \le s_i \le q-1$. This defines the sequence $\ol{s}=(s_{m-1},s_{m-2},\ldots,s_0)$. Denote by $E(s)$ the NDS decomposition
of $\ol{s}$. Conversely, let $\ul{V_1}\ul{V_2}\ldots\ul{V_r}$ be the NDS decomposition of $\ol{s}$. Then define
$$
E^{-1}(\ul{V_1}\ul{V_2}\ldots\ul{V_r})=s.
$$
The coset leader of $C_s$ modulo $n$ is denoted by $s^*$.

When $n=q^m-1$ and $\ol{s}=(s_{m-1},s_{m-2},\ldots,s_0)$, noting that $\ol{n}=(q-1,q-1,\cdots,q-1)$, it is clear that the sequence $\ol{q^is}$ corresponding to $q^is \bmod{n}$ is
\begin{eqnarray} \label{2:xi1}
\ol{q^is}=(s_{m-1-i},\ldots,s_0,s_{m-1},\ldots,s_{m-i}), \  1 \le i \le m-1.
\end{eqnarray}

That is, when $n=q^m-1$, multiplying a power of $q$ simply corresponds to a cyclic shift of the original sequence $\ol{s}$. This key fact is fundamental for the important results presented in \cite{YF}. When $n=\frac{q^m-1}{q-1}$, the situation is more complicated: while (\ref{2:xi1}) is still true, however, when $\ol{q^is}>\ol{n}$, noting that $\ol{n}=(1,1,\ldots,1)$, we need to subtract $\ol{q^is}$ by some multiple of $(1,1,\ldots,1)$ so that the resulting sequence lies in between $\ol{0}$ and $\ol{n}$.
With this observation, we can translate easily some results of \cite{YF} into the new context, where the modulus is $n=\frac{q^m-1}{q-1}$. From now on, we always consider $q$-cyclotomic cosets modulo $n=\frac{q^m-1}{q-1}$.

\begin{lemma}\label{lem-NDS1}
Let $0 \le s \le \frac{q^m-1}{q-1}-1$ be an integer such that the components of $\ol{s}$ are either $0$ or $1$. Suppose $E(s)=\ul{V_1} \ul{V_2} \ldots \ul{V_r}$. Then we have the following.
\begin{itemize}
\item[i)] $E(s^*)=\ul{V_j} \ul{V_{j+1}} \ldots \ul{V_r} \ul{V_1}\ldots\ul{V_{j-1}}$ for some $j$ where $\ul{V_j} \le \ul{V_i}$ for each $1 \le i \le r$.
\item[ii)] If $\ul{V_1}=\ul{V_2}=\cdots=\ul{V_r}$ or $\ul{V_1}=\ul{V_2}=\cdots=\ul{V_j}<\ul{V_k}$ for some $j$ and for all $k>j$, then $s=s^*$.
\item[iii)] If $r=1$, then $s=s^*$.
\end{itemize}
\end{lemma}
\begin{proof}
Since $\ol{n}=(1,1,\cdots,1)$ and the components of $\ol{s}$ are either $0$ or $1$, the components of $\ol{q^is}$ are also either $0$ or $1$. Thus $0<\ol{q^is}<\ol{n}$ for all $i$. Therefore the process of subtracting $\ol{q^is}$ by multiples of $\ol{n}$ is not involved. The proof is exactly the same as that of \cite[Theorem 2.2]{YF} for $n=q^m-1$. \end{proof}

Let $\ul{v}=(v_{l-1},v_{l-2},\ldots,v_0)$ be an NDS. Define the {\rm truncating operator} $T_k$ as
$$
T_k(\ul{v})=(v_{l-1},v_{l-2},\ldots,v_{l-k}), \ 1 \le k \le l.
$$
We now define the {\rm successor operator} $S$ in the way that $S(\ul{v})$ is the smallest NDS that is larger than $\ul{v}$.
In particular, if $v_0<q-1$, we have
$$
S(\ul{v})=(v_{l-1},v_{l-2},\ldots,v_0+1),
$$
which corresponds to the successor of the integer $\sum_{i=0}^{l-1}v_iq^i$. The following Lemma provides information about the coset leaders modulo $\frac{q^m-1}{q-1}$ in a special case.

\begin{lemma}\label{lem-NDS2}
Let $0 \le s \le \frac{q^m-1}{q-1}-1$ be an integer. Suppose $E(s)=\ul{V_1} \ul{V_2} \ldots \ul{V_r}$, where $\ul{V_1}>\ul{V_2}$. Suppose $\ul{V_1}$ has length $l$ and has components either $0$ or $1$.  Let $M(s)$ be the smallest coset leader greater than or equal to $s$. Write $m=al+b$, where $0 \le b \le l-1$. If $b=0$, we have
$$
M(s)=E^{-1}(\underbrace{\ul{V_1}\ul{V_1}\ldots\ul{V_1}}_{a}).
$$
If $1 \le b \le l-1$, we have
$$
M(s) \ge E^{-1}(\underbrace{\ul{V_1}\ul{V_1}\ldots\ul{V_1}}_{a}S(T_b(\ul{V_1}))).
$$
In particular, if the last component of $S(T_b(\ul{V_1}))$ is $1$, then the equality holds.
\end{lemma}
\begin{proof}
The proof is the same as that of \cite[Theorem 2.5]{YF} for $n=q^m-1$, because in (\ref{2:xi1}) no subtracting of $(1,1,\ldots,1)$ from $\ol{q^is}$ is involved.
\end{proof}

\subsection{Gauss and exponential sums related to quadratic forms}

In this subsection, we list two results on Gauss sums and exponential sums related to quadratic forms over finite fields. Interested readers may refer to \cite[Chapter 5]{LN} for other properties of Gauss sums and to \cite[Chapter 6]{LN} and \cite{LF} for the general theory of quadratic forms over finite fields.

From what follows, we assume that $p$ is an odd prime and $q=p^s$. Denote by $\Tr$ the standard trace map from $\Fq$ to $\Fp$.

\begin{definition}
Let $\chi$ be a multiplicative character over $\Fq$. The Gauss sum $G(\chi)$ is defined to be
$$
G(\chi)=\sum_{x \in \Fq} \chi(x)\zeta_p^{\Tr(x)},
$$
where $\zeta_p:=\exp(2 \pi \sqrt{-1}/p)$ is a $p$-th complex root of unity.
\end{definition}

Let $\eta$ be the quadratic character of $\gf(q)$.

\begin{lemma}{\rm \cite[Theorem 5.15]{LN}}\label{lemmagauss} The quadratic Gauss sum $G(\eta)$ satisfies
$$
G(\eta)=\begin{cases}
      (-1)^{s-1}\sqrt{q} & \text{if $p\equiv 1 \bmod 4$},\\
      (-1)^{s-1}(\sq)^s\sqrt{q} & \text{if $p\equiv 3 \bmod 4$}.
\end{cases}
$$
\end{lemma}
The following identity also holds (see \cite[Theorem 5.33]{LN}):
\begin{eqnarray} \label{pri:1}
\sum_{x \in \gf(q)^*} \zeta_p^{\Tr(ax^2)}=\eta(a)G(\eta)-1, \quad \forall \, a \in \gf(q)^*.
\end{eqnarray}
\begin{lemma}{\rm \cite[Lemma 1]{LF}}\label{lemmaquad}
Let $Q(x)$ be a quadratic form from $\gf(q^m)$ to $\gf(q)$ with rank $r$. Then
$$
\sum_{x \in \gf(q^m)} \zeta_p^{\Tr(Q(x))}=
\begin{cases}
\pm q^{m-r/2} & \mbox{if $q\equiv 1 \bmod 4$}, \\
\pm (\sq)^r q^{m-r/2} & \mbox{if $q\equiv 3 \bmod 4$}.
\end{cases}
$$
Moreover, for $a \in \gf(q)^*$, we have
$$
\sum_{x \in \gf(q^m)} \zeta_p^{\Tr(aQ(x))}=\eta(a^r)\sum_{x \in \gf(q^m)} \zeta_p^{\Tr(Q(x))}.
$$
\end{lemma}

\subsection{Some known results concerning BCH codes}

Locator polynomials are very useful in the study of BCH codes \cite{ACS92,AS94}. We first review the definition.

\begin{definition}\label{def-locator}
Let $c=(c_0,c_1,\ldots,c_{n-1}) \in \Fq^n$ be a vector with nonzero components $c_{i_1},c_{i_2},\ldots,c_{i_w}$. Then
$$
X_1=\be^{i_1},\ldots,X_w=\be^{i_w},
$$
are called the {\rm locators} of $c$, where $\beta$ is a primitive $n$-th root of unity in $\gf(q^m)$. The {\rm locator polynomial} of $c$ is
\begin{align*}
\sig(z) = \prod_{i=1}^w(1-X_iz)
        = \sum_{i=0}^w \sig_i z^i,
\end{align*}
where $\sig_0=1$. The coefficients $\sig_i$ are the elementary symmetric functions of $X_i$:
\begin{align*}
\sig_1&=-(X_1+\cdots+X_w),\\
\sig_2&=X_1X_2+X_1X_3+\cdots+X_{w-1}X_w,\\
      &\vdots\\
\sig_w&=(-1)^wX_1 \cdots X_w.
\end{align*}
\end{definition}

The following lemma suggests a method to find a codeword of prescribed weight in the BCH code $\C_{(n, q, m, \de)}$ by using locator polynomials.

\begin{lemma}{\rm \cite[Ch. 9, Lemma 4]{MS77}}\label{lem-locator}
Let
$$
\sig(z)=\sum_{i=0}^w \sig_iz^i
$$
be a polynomial over $\Fqm$. Then $\sig(z)$ is the locator polynomial of a codeword $c \in \C_{(n, q, m, \de)}$ with only $0$ and $1$ components if and only if the following two conditions hold.
\begin{itemize}
\item[i)] The zeroes of $\sig(z)$ are distinct $n$-th roots of unity.
\item[ii)] $\sig_i=0$ for all $1 \le i \le \de-1$ with $p \nmid i$, where $p$ is the characteristic of $\Fq$.
\end{itemize}
\end{lemma}

The following lemma says in some cases, the minimum distance equals the designed distance.

\begin{lemma}{\rm \cite[Theorem 4.3.13]{BBFKKW}}\label{lem-mineqdes}
For the BCH code $\C_{(n, q, m, \de)}$, if $\de \mid n$, then the minimum distance $d=\de$.
\end{lemma}

\section{The narrow-sense BCH codes with large dimensions}\label{sec-largedim}

For the rest of this paper, unless otherwise stated, we will always assume that $n=\frac{q^m-1}{q-1}$. Therefore $\ord_n(q)=m$. We use $\al$ to denote a primitive element of $\Fqm$ and $\be=\al^{q-1}$. In this section, we consider the narrow-sense BCH codes $\C_{(n,q,m,\delta)}$ and $\tilde{\C}_{(n,q,m,\delta)}$ with few zeroes, and thus they have large dimensions. The narrow-sense BCH code with designed distance $2$ has only one zero. The parameters of these codes are known.


\begin{theorem}\label{thm-delta001}
The code $\C_{(n, q, m, 2)}$ has parameters $[(q^m-1)/(q-1), (q^m-1)/(q-1)-m, d]$, where
$$
d=\begin{cases}
  3 & \text{if $\gcd(m,q-1)=1$,} \\
  2 & \text{if $\gcd(m,q-1)> 1$.}
\end{cases}
$$
\end{theorem}
\begin{proof}
The dimension follows from the fact that $|C_1|=m$. Since $\C_{(n, q, m, 2)}$ has only one zero $\beta$, its parity-check matrix is
$$
H=(1,\be,\ldots,\be^{n-1}),
$$
where the $i$-th column is a vector in $\Fq^m$ corresponding to $\be^{i-1}$. Notice that
\[\gcd(n,q-1)=\gcd(q^{m-1}+q^{m-2}+\cdots+q+1,q-1) =\gcd(m,q-1).\]
If $\gcd(m,q-1)=1$, then $\gcd(n,q-1)=1$. Suppose $H$ contains two columns which are linearly dependent over $\gf(q)$. Let these two columns be the $i$-th column and the $j$-th column where $1 \le i < j \le n$. Then we have $\be^{j-1}=\theta\be^{i-1}$ for some $\theta \in \gf(q)^*$. Thus $\be^{j-i}=\theta \in \Fq^*$, which implies that $\be^{(j-i)(q-1)}=1$. Since $\beta$ is a primitive $n$-th root of unity and $\gcd(n,q-1)=1$, we have $n|(j-i)$. This is impossible since $1 \le j-i \le n-1$. Therefore every two distinct columns of $H$ are linearly independent over $\gf(q)$ if $\gcd(m,q-1)=1$. In this case the code $\C_{(n, q, m, 2)}$ is the Hamming code and we have $d=3$ \cite[p. 29-30]{HP03}. If $\gcd(m, q-1) =l> 1$, from the above argument it is easy to see that the first column and the $(\frac{n}{l}+1)$-th column of $H$ are linearly dependent over $\gf(q)$. That is, the code $\C_{(n, q, m, 2)}$ has minimum distance $d=2$.
\end{proof}

\begin{theorem}\label{thm-delta002}
The code $\tilde{\C}_{(n, q, m, 2)}$ has parameters $[(q^m-1)/(q-1), (q^m-1)/(q-1)-m-1, d]$, where $3 \le d \le 4$.
\end{theorem}
\begin{proof}
$\tilde{\C}_{(n, q, m, 2)}$ is the even-like subcode of $\C_{(n, q, m, 2)}$ with dimension $(q^m-1)/(q-1)-m-1$. The minimum distance $2 \le d \le 4$, where the lower bound follows from the BCH bound and the upper bound follows from the sphere packing bound. Note that the parity-check matrix is
$$
\widetilde{H}=\begin{pmatrix}
 1 & 1 & 1 & \cdots & 1 \\
 &&&&\\
   &   & H &        &   \\
 &&&&
 \end{pmatrix},
$$
where $H=(1,\be,\ldots,\be^{n-1})$ is the parity-check matrix of the code $\C_{(n,q,m,2)}$ in Theorem~\ref{thm-delta001}. We claim that every two columns of $\widetilde{H}$ are linearly independent over $\gf(q)$. For any $1 \le i < j \le n$, assume that the $i$-th and $j$-th columns are linearly dependent over $\gf(q)$. Then we must have $\be^{i-1}=\be^{j-1}$, that is, $\be^{j-i}=1$. This is impossible since $\be$ is a primitive $n$-th root of unity and $1 \le j-i \le n-1$. Hence the claim is proved. Therefore the minimum distance $d$ of $\tilde{\C}_{(n, q, m, 2)}$ is strictly larger than two and we have $3 \le d \le 4$.
\end{proof}

For the narrow-sense BCH codes with designed distance $3$, the parameters can be determined in some cases and are described as follows.

\begin{theorem}\label{thm-delta0031}
Let $q \ge 3$. The code $\C_{(n, q, m, 3)}$ has parameters $[(q^m-1)/(q-1), (q^m-1)/(q-1)-2m, d]$, where $d$ can be determined in the following cases:

(1). If
\begin{itemize}
\item[i)] $q \equiv 1 \bmod 3$ and $3 \mid m$, or
\item[ii)] $q \equiv 2 \bmod 3$ and $2 \mid m$,
\end{itemize}
then $d=3$.

(2). If
\begin{itemize}
\item[iii)] $q \equiv 1 \bmod 4$ and $4 \mid m$, or
\item[iv)] $q \equiv 3 \bmod 4$ and $2 \mid m$,
\end{itemize}
then $3 \le d \le 4$.

(3). If $q=3$ and $2 \mid m$, then $d=4$.
\end{theorem}
\begin{proof}
The dimension of $\C_{(n, q, m, 3)}$ follows from the fact that $|C_1|=|C_2|=m$. By the BCH bound, the minimum distance $d$ is at least three.

We first prove $d=3$ when i) or ii) holds. Since $\de=3$, by Lemma~\ref{lem-locator}, it suffices to find a polynomial of the form $\sig(z)=\sum_{i=0}^3 \sig_iz^i$ such that all roots of $\sig(z)$ belong to the cyclic group $\langle \be \rangle$ and $\sig_1=\sig_2=0$. Notice that both i) and ii) imply that $3 \mid \frac{q^m-1}{q-1}$. Taking $X_1=1,X_2=\be^{\frac{q^m-1}{3(q-1)}},X_3=\be^{\frac{2(q^m-1)}{3(q-1)}}$, we find that $X_1,X_2,X_3 \in \langle \be \rangle$ are distinct and satisfy
$$
\left\{\begin{array}{c}
           X_1+X_2+X_3=0, \\
           X_1X_2+X_1X_3+X_2X_3=0,
       \end{array}\right.
$$
Thus the desired polynomial can be chosen as $\sig(z)=\prod_{i=1}^3 (1-X_iz)$. Hence, there exists a codeword of weight three and we have $d=3$.

Next, we prove $3 \le d \le 4$ when iii) or iv) holds. By Lemma~\ref{lem-locator}, it suffices to find a polynomial of the form $\sig^{\prime}(z)=\sum_{i=0}^4 \sig_i^{\prime}z^i$ such that all roots of $\sig^{\prime}(z)$ belong to the cyclic group $\langle \be \rangle$ and $\sig_1^{\prime}=\sig_2^{\prime}=0$. Notice that both iii) and iv) imply that $4 \mid \frac{q^m-1}{q-1}$ and $-1 \in \langle \be \rangle$. Thus taking $X_1^{\prime}=1,X_2^{\prime}=-1,X_3^{\prime}=\be^{\frac{q^m-1}{4(q-1)}},
X_4^{\prime}=-\be^{\frac{q^m-1}{4(q-1)}}$, we find that $X_1^{\prime},X_2^{\prime},X_3^{\prime},X_4^{\prime} \in \langle \be \rangle$ are distinct and satisfy
$$
\left\{\begin{array}{c}
           X_1^{\prime}+X_2^{\prime}+X_3^{\prime}+X_4^{\prime}=0, \\
           X_1^{\prime}X_2^{\prime}+X_1^{\prime}X_3^{\prime}+X_1^{\prime}X_4^{\prime}+
           X_2^{\prime}X_3^{\prime}+X_2^{\prime}X_4^{\prime}+X_3^{\prime}X_4^{\prime}=0,
       \end{array}\right.
$$
Thus the desired polynomial can be chosen as $\sig^{\prime}(z)=\prod_{i=1}^4 (1-X_i^{\prime}z)$. Hence, there exists a codeword of weight four and we have $3 \le d \le 4$.

As for (3), if $q=3$, we have $d \ge 4$ by the BCH bound. On the other hand, when $q=3$ and $2 \mid m$, we have shown there is a codeword of weight four. Therefore we have $d=4$.
\end{proof}

We make a remark here. Theoretically, the sphere packing bound gives the following restriction on the parameters of $\C_{(n, q, m, 3)}$:
$$
\sum_{i=0}^{\lfloor \frac{d-1}{2} \rfloor} (q-1)^i \binom{n}{i} \le q^{2m}.
$$
Suppose $m \ge 2$ and $d \ge 7$, then the above inequality leads to
$$
q^m+\frac{(q^m-1)(q^m-q)}{2}+\frac{(q^m-1)(q^m-q)(q^m-2q+1)}{6} \le q^{2m},
$$
After a direct simplification, we have
$$
q^{2m-1}-3q^{m}-2q^{m-1}+2q-4 \le 0.
$$
When $m \ge 2$, this inequality does not hold, except for $(q,m)=(3,2)$. In addition, by Theorem~\ref{thm-delta0031}, the code has minimum distance four when $(q,m)=(3,2)$. Hence, for $q \ge 3$ and $m \ge 2$, we have shown $\C_{(n, q, m, 3)}$ has minimum distance $d \le 6$. Together with the BCH bound, we obtain the following restriction on the minimum distance $d$ of $\C_{(n, q, m, 3)}$, where $m \ge 2$:
$$
\begin{cases}
  4 \le d \le 6  & \mbox{if $q=3$,} \\
  3 \le d \le 6  & \mbox{if $q>3$.}
\end{cases}
$$
On the other hand, Magma examples (for $3 \leq m \leq 6$ and $3 \leq q \leq 5$) suggest that $d=4$ when $q=3$ and $3 \le d \le 4$ when $q>3$. These experimental results have been partially explained in Theorem~\ref{thm-delta0031}. However, we are not sure if the method employing locator polynomials can be applied to the remaining cases, since this approach can only find the codeword whose components are either $0$ or $1$.

\begin{example}
The code $\C_{(n, q, m, 3)}$ has parameters $[40,32,4]$ when $(q, m)=(3,4)$,
$[1365,1353,3]$ when $(q, m)=(4,6)$, and
$[156,148,3]$ when $(q, m)=(5,4)$. The ternary code with parameters $[40,32,4]$
is an optimal cyclic code according to {\rm \cite[p. 313]{Dingbk15}}.
\end{example}

\section{The narrow-sense BCH codes with small dimensions}\label{sec-smalldim}

In this section, we study narrow-sense BCH codes with small dimensions. Our task is to find the first few largest coset leaders modulo $n=(q^m-1)/(q-1)$. By Proposition~\ref{prop-Bose}, the Bose distance of a narrow-sense BCH code must be a coset leader. The knowledge of these coset leaders provides information on the Bose distance and dimension of narrow-sense BCH codes whose zeroes include all roots of $x^n-1$ except those corresponding to the first few largest coset leaders. We denote the first and the second largest coset leader modulo $n$ by $\de_1$ and $\de_2$, respectively. It seems a hard problem to determine
$\delta_1$ and $\delta_2$ for all $q > 2$. For the rest of this section, we only deal with the case $q=3$.

\subsection{The two coset leaders $\de_1$ and $\de_2$}

\begin{lemma}\label{lem-delta1}
Let $q=3$ and $m \ge 2$. The first largest coset leader modulo $n=(q^m-1)/(q-1)$ is
$$
\delta_1=q^{m-1}-1-\frac{q^{\lfloor (m-1)/2 \rfloor}-1}{q-1}
$$
and
\begin{eqnarray*}
|C_{\delta_1}|=\left\{ \begin{array}{ll}
m & \mbox{ if $m$ is odd, } \\
\frac{m}{2}  & \mbox{ if $m$ is even.}

\end{array}
\right.
\end{eqnarray*}
The second largest coset leader modulo $n$ is
$$
\delta_2=q^{m-1}-1-\frac{q^{\lfloor (m+1)/2 \rfloor}-1}{q-1}
$$
and $|C_{\delta_2}|=m$.
\end{lemma}

\begin{proof}
When $2 \le m \le 3$, the desired conclusions can be verified directly. Below, we consider the case that $m \ge 4$.
Suppose $0< \de <n$ is a coset leader of the form $\ol{\de}=(a_{m-1},a_{m-2},\ldots,a_0)$. Since $\de <n$, we have $0 \le a_{m-1} \le 1$. We first observe that $a_{m-1}=0$. Otherwise, assume $a_{m-1} =1$. Then, there is a component $a_i$ such that $a_i=0$. We can take a cyclic shift $\ol{q^j \de}$ for some $j$ (see (\ref{2:xi1})) so that $a_i=0$ becomes the first component. This implies $\ol{q^j \de} <\ol{\de}$, which is a contradiction to the assumption that $\de$ is a coset leader.

Next we assume that the coset leader $\de$ is of the form $\ol{\de}=(0,2,a_{m-3},\ldots,a_0)$. By the same argument as before, $00$ and $01$ cannot appear in the sequence $\ol{\de}$. Moreover, if $12$ appears, by taking a suitable cyclic shift we have $\ol{q^j\de}=(1,2,\ldots, 0,2,\ldots)$ for some $j$. It is easy to see that
\[0<\ol{q^j\de}-\ol{n}=(1,2,\ldots, 0,2,\ldots)-(1,1,\ldots,1)=(0,b_{m-1},\ldots,b_0)=\ol{v},\]
where
$0\le b_{m-1} \le 1$, hence $\ol{v}<\ol{\de}$, a contradiction to the assumption that $\de$ is a coset leader. So $12$ does not appear in $\ol{\de}$.

Thirdly, let
\begin{eqnarray} \label{2:xi2} \ol{\de}=(0,\underbrace{2,\ldots,2}_{u},\underbrace{1,\ldots,1}_{v}),\end{eqnarray} where $u,v \ge 1$ and $u+v+1=m$. We can check that the sequences corresponding to $q^i \de \bmod{n}$ are given by
\begin{align*}
\ol{q^i\de}&=(\underbrace{1,\ldots,1}_{u-i},0,\underbrace{2,\ldots,2}_{v+1},\underbrace{1,\ldots,1}_{i-1}), \ 1 \le i \le u, \\
\ol{q^{u+1+i}\de}&=(\underbrace{1,\ldots,1}_{v-i},0,\underbrace{2,\ldots,2}_{u},\underbrace{1,\ldots,1}_{i}), \ 0 \le i \le v-1.
\end{align*}
Therefore, the $\de$ of the form (\ref{2:xi2}) is a coset leader modulo $n$ if and only if $u \le v+1$, that is, $u \le \frac{m}{2}$.

Finally, let $\de$ be a coset leader of the form $\ol{\de}=(0,2,a_{m-3},\ldots,a_{0})$ but not of the form (\ref{2:xi2}). From what we have proved, $\ol{\de}$ must be of the form
\begin{equation}\label{2:xi3}
\ol{\de}=(0,\underbrace{2,\ldots,2}_{u_1},\underbrace{1,\ldots,1}_{v_1},
0,\underbrace{2,\ldots,2}_{u_2},\underbrace{1,\ldots,1}_{v_2},
\cdots,0,\underbrace{2,\ldots,2}_{u_t},\underbrace{1,\ldots,1}_{v_t}),
\end{equation}
where $t \ge 2, u_i \ge 1, v_i \ge 0 $ for each $i$ and $u_1 \le u_i$ for each $2 \le i \le t$. In particular, we have $u_1 \le u_2$ and $u_1+u_2+2 \le m$, which implies that $u_1 \le \frac{m}{2}-1$.

From the argument above, we conclude that the largest two coset leaders $\de_1$ and $\de_2$ are given as below:
\begin{align*}
\ol{\de_1}=(0,\underbrace{2,\ldots,2}_{\frac{m-1}{2}},
\underbrace{1,\ldots,1}_{\frac{m-1}{2}}),
\ \ol{\de_2}=(0,
\underbrace{2,\ldots,2}_{\frac{m-3}{2}},
\underbrace{1,\ldots,1}_{\frac{m+1}{2}}), & \mbox{ if } m \ge 5 \mbox{ is odd;} \\
\ol{\de_1}=(0,\underbrace{2,\ldots,2}_{\frac{m}{2}},
\underbrace{1,\ldots,1}_{\frac{m}{2}-1}),
\ \ol{\de_2}=(0,\underbrace{2,\ldots,2}_{\frac{m}{2}-1},
\underbrace{1,\ldots,1}_{\frac{m}{2}}), &\mbox{ if } m \ge 4 \mbox{ is even.}
\end{align*}
Consequently, noting $q=3$, we have
\begin{spacing}{1.4}
\begin{align*}
\delta_1 &= \begin{cases}
               2\sum_{i=\frac{m-1}{2}}^{m-2} q^i+\sum_{i=0}^{\frac{m-3}{2}} q^i & \mbox{if $m \ge 5$ is odd} \\
               2\sum_{i=\frac{m}{2}-1}^{m-2} q^i+\sum_{i=0}^{\frac{m}{2}-2} q^i & \mbox{if $m \ge 4$ is even}
            \end{cases} \\
         &=
           \begin{cases}
               q^{m-1}-q^{\frac{m-1}{2}}+\frac{q^{\frac{m-1}{2}}-1}{q-1} & \mbox{if $m \ge 5$ is odd} \\
               q^{m-1}-q^{\frac{m}{2}-1}+\frac{q^{\frac{m}{2}-1}-1}{q-1} & \mbox{if $m \ge 4$ is even}
            \end{cases} \\
         &= q^{m-1}-1-\frac{q^{\lfloor (m-1)/2 \rfloor}-1}{q-1}.
\end{align*}
\end{spacing}
\noindent Since $q^{\frac{m}{2}}\de_1 \equiv \de_1 \pmod{n}$ when $m$ is even, we have
\begin{eqnarray*}
|C_{\delta_1}|=\left\{ \begin{array}{ll}
m & \mbox{ if $m$ is odd, } \\
\frac{m}{2}  & \mbox{ if $m$ is even.}
\end{array}
\right.
\end{eqnarray*}
As for $\de_2$ we have
\begin{spacing}{1.4}
\begin{align*}
\delta_2 &= \begin{cases}
               2\sum_{i=\frac{m+1}{2}}^{m-2} q^i+\sum_{i=0}^{\frac{m-1}{2}} q^i & \mbox{if $m \ge 5$ is odd} \\
               2\sum_{i=\frac{m}{2}}^{m-2} q^i+\sum_{i=0}^{\frac{m}{2}-1} q^i & \mbox{if $m \ge 4$ is even}
            \end{cases} \\
         &= \begin{cases}
               q^{m-1}-q^{\frac{m+1}{2}}+\frac{q^{\frac{m+1}{2}}-1}{q-1} & \mbox{if $m \ge 5$ is odd} \\
               q^{m-1}-q^{\frac{m}{2}}+\frac{q^{\frac{m}{2}}-1}{q-1} & \mbox{if $m \ge 4$ is even}
            \end{cases} \\
         &=q^{m-1}-1-\frac{q^{\lfloor (m+1)/2 \rfloor}-1}{q-1}.
\end{align*}
\end{spacing}
It can also be shown that $|C_{\de_2}|=m$. This completes the proof of Lemma \ref{lem-delta1}.
\end{proof}

\subsection{Some other coset leaders $\delta_i$}\label{sec-otherdeltai}

Let $\delta_i$ denote the $i$-th largest coset leader modulo $n=(q^m-1)/(q-1)$. In this subsection, we point out that some of the coset leaders $\delta_i$ can also be determined in the case $q=3$.

Let $\de$ be a coset leader, then according to the proof of Lemma~\ref{lem-delta1}, $\ol{\de}$ must have the form (\ref{2:xi2}) or (\ref{2:xi3}). Suppose $\ol{\de}$ has the form (\ref{2:xi3}), namely,
\[\ol{\de}=(0,\underbrace{2,\ldots,2}_{u_1},
\underbrace{1,\ldots,1}_{v_1},\ldots,0,\underbrace{2,\ldots,2}_{u_t},
\underbrace{1,\ldots,1}_{v_t}),\]
where $t \ge 2$, $u_i \ge 1,v_i \ge 0$ for each $1 \le i \le t$ and $\sum_{i=1}^t (u_i+v_i+1)=m$. Since $\de$ is a coset leader, we must have $u_1 \le u_i$ for each $2 \le i \le t$. Moreover, we have
\begin{align*}
\ol{q\de}&=(\underbrace{1,\ldots,1}_{u_1-1},0,\underbrace{2,\ldots,2}_{v_1+1},\underbrace{1,\ldots,1}_{u_2-1},\ldots,0,\underbrace{2,\ldots,2}_{v_t+1}),\\
\ol{q^{u_1}\de}&=(0,\underbrace{2,\ldots,2}_{v_1+1},\underbrace{1,\ldots,1}_{u_2-1},\ldots,0,\underbrace{2,\ldots,2}_{v_t+1},\underbrace{1,\ldots,1}_{u_1-1}).
\end{align*}
By the last equation, $\de$ is a coset leader only if $u_1 \le v_i+1$ for each $1 \le i \le t$. Since $\sum_{i=1}^t (u_i+v_i+1)=m$, we have $u_1+v_1+u_2+v_2+2 \le m$, which implies that $u_1 \le \lfloor \frac{m}{4} \rfloor$. This gives a strong restriction on a coset leader of the form (\ref{2:xi3}). On the other hand, for $1 \le i \le \lceil \frac{m-1}{2} \rceil$, define $\de_i^{\prime}$ to be the following coset leaders of the form (\ref{2:xi2}):
$$
\ol{\de_i^{\prime}}=(0,\underbrace{2,\ldots,2}_{\lceil \frac{m+1}{2} \rceil-i},\underbrace{1,\ldots,1}_{m-1-\lceil \frac{m+1}{2} \rceil+i}).
$$
Since $u_1 \le \lfloor \frac{m}{4} \rfloor$, when $\lceil \frac{m+1}{2} \rceil-i \ge \lfloor \frac{m}{4} \rfloor$, $\de_i^{\prime}$ is greater than any coset leader of the form (\ref{2:xi3}). Thus, for $1 \le i \le \lceil \frac{m}{4} \rceil$, $\de_i$ must be of the form (\ref{2:xi2}). Indeed, we have $\de_i=\de_i^{\prime}$ when $1 \le i \le \lceil \frac{m}{4} \rceil$. Therefore
$$
\de_i=q^{m-1}-1-\frac{q^{\lfloor \mmt+i \rfloor}-1}{q-1}, \quad 1 \le i \le \lceil \frac{m}{4} \rceil.
$$
We can see the condition $q=3$ plays an essential role in the proof. For the case $q>3$, we do not have a similar result.

\subsection{The ternary codes $\C_{(n,3,m,\delta_1)}$ and $\tilde{\C}_{(n,3,m,\delta_1)}$}

In this subsection, we study the ternary codes $\C_{(n,3,m,\delta_1)}$ and $\tilde{\C}_{(n,3,m,\delta_1)}$. We recall the following trace representations of cyclic codes, which is a direct consequence of Delsarte's Theorem \cite{Del}. Note that the length $n$ in the following Proposition is not necessarily of the form $\frac{q^m-1}{q-1}$.

\begin{proposition}\label{prop-trace}
Let $q$ be a prime power and $m=\ord_n(q)$. Let $\gamma$ be an $n$-th primitive root of unity in $\gf(q^m)$ and $\mathcal{C}$ be a cyclic code of length $n$ over $\gf(q)$. Suppose $\mathcal{C}$ has $s$ nonzeroes and let $\gamma^{i_1}, \gamma^{i_2}, \ldots, \gamma^{i_s}$ be the $s$ roots of its parity-check polynomial which are not conjugate with each other. Denote the size of the $q$-cyclotomic coset $C_{i_j}$ to be $m_j$, $1 \le j \le s$. Then $\mathcal{C}$ has the following trace representation
$$
\mathcal{C}=\{ c(a_1,a_2,\ldots,a_s) \mid a_j \in \gf(q^{m_j}), 1 \le j \le s \},
$$
where
$$
c(a_1,a_2,\ldots,a_s)=(\sum_{j=1}^s \Tr^{q^{m_j}}_{q}(a_j\gamma^{-li_j}))_{l=0}^{n-1}.
$$
\end{proposition}

Given a codeword $c$, we use $w(c)$ to denote the Hamming weight of $c$. For the ternary code $\C_{(n,3,m,\delta_1)}$, we have the following theorem.

\begin{theorem}\label{thm-delta1}
Let $m \geq 3$. Then the ternary code $\C_{(n,3,m,\delta_1)}$ has parameters
$$
\left[ \frac{3^m-1}{2},\, k,\,   \delta_1 \right],
$$
where
\begin{eqnarray*}
k=\left\{ \begin{array}{ll}
m+1 & \mbox{ if $m$ is odd, } \\
\frac{m+2}{2}  & \mbox{ if $m$ is even.}
\end{array}
\right.
\end{eqnarray*}
In addition, $\C_{(n,3,m,\de_1)}$ is a three-weight code if $m \geq 4$ is even, and a four-weight code if $m \geq 3$ is odd. The weight distribution of $\C_{(n, 3, m, \delta_1)}$ is listed in Table~\ref{table1} and Table~\ref{table2}.
\end{theorem}

\begin{proof}
The conclusion on the dimension of the code follows from Lemma \ref{lem-delta1}. As for the weight distribution, we consider only the case that $m \ge 3$ is odd. The case that $m \ge 4$ is even can be treated in the same way and we omit the details.

Note that $\C_{(n,3,m,\de_1)}$ has two nonzeroes. More precisely, $1$ and $\be^{\de_1}$ are two non-conjugate roots of its parity-check polynomial. Since
\begin{align*}
\frac{3^m-1}{2}-\de_1 &= \frac{3^m-1}{2}-3^{m-1}+1+\frac{3^{\frac{m-1}{2}}-1}{2} \\ &=\frac{3^{m-1}+3^{\frac{m-1}{2}}}{2},
\end{align*}
by Proposition~\ref{prop-trace}, we have
$$
\C_{(n,3,m,\de_1)}=\{ c_1(a,b) : a \in \gf(3^m), \ b \in \gf(3)\},
$$
where
\begin{align*}
c_1(a,b)=&\left(\Trtmt\left(a\be^{\frac{(3^{m-1}+3^{\frac{m-1}{2}})j}{2}}\right)+b\right)_{j=0}^{n-1} \\
              =&\left(\Trtmt\left(a\al^{(3^{m-1}+3^{\frac{m-1}{2}})j}\right)+b\right)_{j=0}^{n-1}.
\end{align*}
Here $\alpha$ is a primitive element of $\gf(3^m)$ and $\be=\al^2$. Note that
\begin{align*}
\gcd(3^{m-1}+3^{\mmo},3^m-1)&=\gcd(3^{m-1}+3^{\mmo},3^{\mmo}+1) \\
                            &=\gcd(3^{\mmo}+1,3^{\mpo}+1) \\
                            &=\gcd(3^{\mmo}+1,2) \\
                            &=2.
\end{align*}
Thus the code $\C_{(n,3,m,\de_1)}$ is permutation equivalent to the following code
\begin{eqnarray} \label{2:xi4}
\left\{c_2(a,b)=\left(\Trtmt\left(a\al^{2j}\right)+b\right)_{j=0}^{n-1}: \quad a \in \gf(3^m), \ b \in \gf(3)\right\}.
\end{eqnarray}
Because of the simple expression of the code (\ref{2:xi4}) and because the weight distribution problem has been studied extensively for many families of cyclic codes in recent years (see for example \cite{DLLZ} for a recent update and references therein), the weight distribution of the code in (\ref{2:xi4}) may be known, but we have not found it in the literature. For the sake of completeness and the illustration of ideas and methods which will be used later, we include a detailed computation below.

First, if $a=0$ and $0 \ne b \in \gf(3)$, it is easy to see that the Hamming weight $w(c_2(a,b))$ of the codeword $c_2(a,b)$ takes the value $\frac{3^m-1}{2}$ twice.

Second, if $0 \ne a \in \gf(3^m)$ and $b \in \gf(3)$, we have
\begin{align*}
w(c_2(a,b))&=n-|\{0 \le j \le n-1 \mid \Trtmt(a\al^{2j})+b=0 \}| \\
         &=n -\sum_{j=0}^{n-1}\frac13 \sum_{x \in \Ft} \z^{x(\Trtmt(a\al^{2j})+b)} \\
         &=\frac23 n -\frac13 \sum_{x \in \Ft^*} \z^{bx} \sum_{j=0}^{n-1} \z^{\Trtmt(ax\al^{2j})} \\
         &=\frac23 n -\frac16 \sum_{x \in \Ft^*} \z^{bx} \sum_{y \in \Ftm^*} \z^{\Trtmt(axy^2)}.
\end{align*}
By (\ref{pri:1}) we obtain
\begin{align*}
w(c_2(a,b))&=\frac23 n -\frac16 \sum_{x \in \Ft^*} \z^{bx} (\eta(ax)G(\eta)-1) \\
&=\frac23 n -\frac16 \eta(a)G(\eta) \left(\z^b+\eta(-1)\z^{-b}\right)+ \frac16 \left(\z^b+\z^{-b}\right), \end{align*}
where $\eta$ is the quadratic character of $\gf(3^m)$ and $G(\eta)$ is the quadratic Gauss sum over $\Ftm$.

Since $m$ is odd, $\eta(-1)=-1$. Using the values of $G(\eta)$ from Lemma \ref{lemmagauss}, we find that if $a\ne 0$ and $b=0$, $w(c_2(a,b))$ takes the value $3^{m-1}$ for $3^m-1$ times, and if $a \ne 0$ and $b \ne 0$, the weight $w(c_2(a,b))$ takes each of the values $3^{m-1}-\frac{1+(-1)^{(m+1)/2}3^{(m-1)/2}}{2}$ and $3^{m-1}-\frac{1+(-1)^{(m-1)/2}3^{(m-1)/2}}{2}$ for $3^m-1$ times. Therefore, we obtain the weight distribution for the case that $m \ge 3$ is odd, which is recorded in Table~\ref{table2}.
\end{proof}

\begin{table}
\ra{1.9}
\begin{center}
\caption{The weight distribution of $\C_{(n,q,m,\de_1)}$ when $m\ge4$ is even}
\begin{tabular}{|c|c|}
\hline
Weight  &  Frequency \\ \hline
$0$     &  $1$ \\ \hline
$3^{m-1}-\frac{3^{m/2-1}+1}{2}$ & $2(3^{m/2}-1)$ \\ \hline
$3^{m-1}+3^{m/2-1}$ & $3^{m/2}-1$ \\ \hline
$\frac{3^m-1}{2}$ & $2$ \\ \hline
\end{tabular}
\label{table1}
\end{center}
\end{table}

\begin{table}
\ra{1.9}
\begin{center}
\caption{The weight distribution of $\C_{(n,q,m,\de_1)}$ when $m\ge3$ is odd}
\begin{tabular}{|c|c|}
\hline
Weight  &  Frequency \\ \hline
$0$     &  $1$ \\ \hline
$3^{m-1}$ & $3^m-1$ \\ \hline
$3^{m-1}-\frac{1+(-1)^{(m-1)/2}3^{(m-1)/2}}{2}$ & $3^m-1$ \\ \hline
$3^{m-1}-\frac{1+(-1)^{(m+1)/2}3^{(m-1)/2}}{2}$ & $3^m-1$ \\ \hline
$\frac{3^m-1}{2}$ & $2$ \\ \hline
\end{tabular}
\label{table2}
\end{center}
\end{table}

\begin{example}
Let $(q, m)=(3,4)$. Then the code $\C_{(n,q, m, \delta_1)}$ of Theorem \ref{thm-delta1} has parameters
$[40, 3, 25]$, and weight enumerator $1+16z^{25}+8z^{30}+2z^{40}$. This code is the best cyclic code
according to {\rm \cite[p. 305]{Dingbk15}}.
\end{example}

\begin{example}
Let $(q, m)=(3,5)$. Then the code $\C_{(n, q, m, \delta_1)}$ of Theorem \ref{thm-delta1} has parameters
$[121, 6, 76]$, and weight enumerator $1+242z^{76}+242z^{81}+242z^{85}+2z^{121}$.
\end{example}

\begin{theorem}\label{thm-delta11}
Let $m \geq 3$. Then the ternary code  $\tilde{\C}_{(n, 3, m, \delta_1)}$ has parameters
$$
\left[ \frac{3^m-1}{2},\, k,\, d \right],
$$
where
\begin{eqnarray*}
k=\left\{ \begin{array}{ll}
m & \mbox{ if $m$ is odd, } \\
\frac{m}{2}  & \mbox{ if $m$ is even,}
\end{array}
\right.
\end{eqnarray*}
and
\begin{eqnarray*}
d=\left\{ \begin{array}{ll}
3^{m-1} & \mbox{ if $m$ is odd, } \\
3^{m-1}+3^{\frac{m}{2}-1}  & \mbox{ if $m$ is even.}
\end{array}
\right.
\end{eqnarray*}
In addition, $\tilde{\C}_{(n, 3, m, \delta_1)}$ has only one nonzero weight.
\end{theorem}
\begin{proof}
The conclusion on the dimension of the code follows from Lemma \ref{lem-delta1}. Note that $\tilde{\C}_{(n,3,m,\de_1)}$ has one nonzero and $\be^{\de_1}$ is a root of its parity-check polynomial. By using Proposition~\ref{prop-trace}, we have
$$
\tilde{\C}_{(n,3,m,\de_1)}=\{c_1(a) : a \in \gf(3^m)\},
$$
where
\begin{align*}
c_1(a)&=\left(\Trtmt\left(a\be^{\frac{(3^{m-1}+3^{\lfloor \frac{m-1}{2} \rfloor})j}{2}}\right)\right)_{j=0}^{n-1} \\
    &=\left(\Trtmt\left(a\al^{(3^{m-1}+3^{\lfloor \frac{m-1}{2} \rfloor})j}\right)\right)_{j=0}^{n-1}.
\end{align*}
We consider only the case that $m \ge 3$ is odd. The case that $m \ge 4$ is even can be treated in the same way. In this case, as was shown in the proof of Theorem~\ref{thm-delta1}, since $\gcd(3^{m-1}+3^{\mmo},3^m-1)=2$, $\tilde{\C}_{(n,3,m,\de_1)}$ is permutation equivalent to the following code
\begin{eqnarray}\label{2:xi5}
\left\{c_2(a)=\left(\Trtmt\left(a\al^{2j}\right)\right)_{j=0}^{n-1}, \quad a \in \gf(3^m)\right\}.
\end{eqnarray}
Again the weight distribution of the code of (\ref{2:xi5}) may be known, but we have not found it in the literature. We provide details of the computation as below.

Clearly, $c_2(0)$ is the zero codeword. If $0 \ne a \in \gf(3^m)$, using (\ref{pri:1}), we have
\begin{align*}
w(c_2(a))&=n-|\{0 \le j \le n-1 \mid \Trtmt(a\al^{2j})=0 \}| \\
         &=n -\sum_{j=0}^{n-1}\frac13 \sum_{x \in \Ft} \z^{x(\Trtmt(a\al^{2j}))} \\
         &=\frac23 n -\frac13 \sum_{x \in \Ft^*} \sum_{j=0}^{n-1} \z^{\Trtmt(ax\al^{2j})} \\
         &=\frac23 n -\frac16 \sum_{x \in \Ft^*} \sum_{y \in \Ftm^*} \z^{\Trtmt(axy^2)} \\
         &=\frac23 n -\frac16 \sum_{x \in \Ft^*} (\eta(ax)G(\eta)-1) \\
         &=3^{m-1}-\frac16 G(\eta) \sum_{x \in \Ft^*} \eta(ax),
\end{align*}
where $\eta$ is the quadratic character of $\gf(3^m)$ and $G(\eta)$ is the quadratic Gauss sum over $\Ftm$. Since $\eta(1)=1$ and $\eta(-1)=-1$, we have $\sum_{x \in \Ft^*} \eta(ax)=0$ for each $a\ne0$. Therefore, $w(c_2(a))=3^{m-1}$ for each $a \ne 0$. Namely, $\C_{(n,3,m,\de_1)}$ has only one nonzero weight $3^{m-1}$.
\end{proof}

We remark that, as it could be easily verified, the code $\tilde{\C}_{(n, 3, m, \delta_1)}$ of Theorem \ref{thm-delta11} meets the
Griesmer bound for linear codes, and hence is optimal.



\subsection{The ternary codes $\C_{(n,3,m,\delta_2)}$ and $\tilde{\C}_{(n,3,m,\delta_2)}$}

The determination of the weight distributions of the ternary codes $\C_{(n,3,m,\delta_2)}$ and $\tilde{\C}_{(n,3,m,\delta_2)}$ is more complicated and depends heavily on the theory of quadratic forms over finite fields (see \cite{LN}). We first do some preparations.

For reasons which will be explained later, let $m \ge 3$ be an odd integer. For $a,b \in \Ftm$, define the quadratic form
$$
Q_{a,b}(x)=\Trtmt\left(ax^{3^{\mmo}+1}+bx^{3^{\mmt}+1}\right).
$$
Let $r_{a,b}$ be the rank of the quadratic form $Q_{a,b}$.

\begin{lemma}\label{QF}
Let $m$ be odd, $a,b \in \Ftm$ and $(a,b) \ne (0,0)$. The quadratic form
$$
Q_{a,b}(x)={\rm Tr_3^{3^{m}}}\left(ax^{3^{\mmo}+1}+bx^{3^{\mmt}+1}\right)
$$
has rank $r_{a,b} \in \{m,m-1,m-2,m-3\}$.
\end{lemma}
\begin{proof}
Let $B_{a,b}(x,y)$ be the symmetric bilinear form associated with the quadratic form $Q_{a,b}$, that is,
\begin{align*}
B_{a,b}(x,y)&=Q_{a,b}(x+y)-Q_{a,b}(x)-Q_{a,b}(y)\\
      &=\Trtmt\left(\left(b^{3^{\mpt}}x^{3^{\mpt}}+a^{3^{\mpo}}x^{3^{\mpo}}
      +ax^{3^{\mmo}}+bx^{3^{\mmt}}\right)y\right).
\end{align*}
Recall that the equation
$$
b^{3^{\mpt}}x^{3^{\mpt}}+a^{3^{\mpo}}x^{3^{\mpo}}+ax^{3^{\mmo}}+bx^{3^{\mmt}}=0
$$
has $3^r$ solutions $x \in \Ftm$ if and only if the rank of $Q_{a,b}$ equals $m-r$ (see \cite[p. 81]{ZD}, for instance). Note that the number of solutions of the above equation equals the number of solutions of the following one:
$$
b^{3^{\mpt}}x^{3^{3}}+a^{3^{\mpo}}x^{3^{2}}+ax^{3}+bx=0.
$$
This has at most 27 solutions, thus $r \le 3$ and $r_{a,b} \in \{m,m-1,m-2,m-3\}$.
\end{proof}

For $a,b \in \Ftm$, define
$$
T(a,b)=\sum_{x \in \Ftm} \z^{\Trtmt(ax^{3^{\mmo}+1}+bx^{3^{\mmt}+1})}.
$$
Denote by $\eta_0$ the quadratic character over $\Ft$. Then, by Lemma~\ref{lemmaquad}, we have
\begin{align}
S(a,b) &:=\sum_{y \in \Ft^*} T(ya,yb) \notag \\
       &=\sum_{y \in \Ft^*} \eta_0(y^{r_{a,b}})T(a,b) \notag \\
       &= T(a,b) \left(1+(-1)^{r_{a,b}} \right) \label{eqn-ST}
\end{align}

We will need the following lemma concerning power moment identities.

\begin{lemma}\label{moment}
For the $S(a,b)$ defined above, we have the following:
\begin{itemize}
\item[i)] $\sum_{a,b \in \Ftm} S(a,b)=2 \times 3^{2m}$.
\item[ii)] $\sum_{a,b \in \Ftm} S(a,b)^2=4\times3^{3m}$.
\item[iii)] $\sum_{a,b \in \Ftm} S(a,b)^3=32\times3^{3m}-24\times3^{2m}$.
\item[iv)] $\sum_{a,b \in \Ftm} T(a,b)^2=3^{2m}$.
\end{itemize}
\end{lemma}
\begin{proof}
i) The identity is trivially true.

ii) Define $N_2$ to be the number of pairs $(u,v) \in \Ft^* \times \Ft^*$ and $(x,y) \in \Ftm \times \Ftm$, which satisfy the following two equations:
$$
\begin{cases}
  ux^{3^{\mmo}+1}+vy^{3^{\mmo}+1}=0,  \\
  ux^{3^{\mmt}+1}+vy^{3^{\mmt}+1}=0.
\end{cases}
$$
It is easy to see that $\sum_{a,b \in \Ftm} S(a,b)^2=3^{2m}N_2$. Thus, it suffices to determine $N_2$.

If one of $x,y$ is zero, then the other must also be zero. When $x=y=0$, we have four choices of the pair $(u,v)$. When $x \ne 0$ and $y \ne 0$, the system above is equivalent to
$$
\begin{cases}
  (\frac{x}{y})^{3^{\mmo}+1}=-\frac{v}{u}, \\
  (\frac{x}{y})^{3^{\mmt}+1}=-\frac{v}{u}.
\end{cases}
$$
As shown in the proof of Theorem~\ref{thm-delta1}, we have $\gcd(3^{\mmo}+1,3^m-1)=2$. Similarly,
\begin{align*}
\gcd(3^{\mmt}+1,3^m-1)&=\gcd(3^{\mmt}+1,3^{\mpt}+1) \\
                      &=\gcd(3^{\mmt}+1,3^3-1) \\
                      &=\begin{cases}
                          \gcd(3^{\mmt}-1+2,3^3-1) & \mbox{if $\mmt \equiv 0\bmod3$} \\
                          \gcd(3^{\mmt}-3+4,3^3-1) & \mbox{if $\mmt \equiv 1\bmod3$} \\
                          \gcd(3^{\mmt}-9+10,3^3-1) & \mbox{if $\mmt \equiv 2\bmod3$}
                        \end{cases}\\
                      &=\begin{cases}
                          \gcd(2,3^3-1) & \mbox{if $\mmt \equiv 0\bmod3$} \\
                          \gcd(4,3^3-1) & \mbox{if $\mmt \equiv 1\bmod3$} \\
                          \gcd(10,3^3-1) & \mbox{if $\mmt \equiv 2\bmod3$}
                        \end{cases}\\
                      &=2
\end{align*}
Thus, $-\frac{v}{u} \in \Ft^*$ must be a square. Namely, $-\frac{v}{u}=1$. Thus, we have two choices for the pair $(u,v)$. Meanwhile, $\frac{x}{y}=\pm1$ are precisely all solutions to the following system of equations:
$$
\begin{cases}
  (\frac{x}{y})^{3^{\mmo}+1}=1, \\
  (\frac{x}{y})^{3^{\mmt}+1}=1.
\end{cases}
$$
Thus, we have $2(3^m-1)$ choices for the pair $(x,y)$. To sum up, we have $4(3^m-1)$ choices for the pairs $(u,v)$ and $(x,y)$ when $x \ne 0$ and $y \ne 0$. In total, we have $N_2=4+4(3^m-1)=4\times3^m$.

iii) Define $N_3$ to be the number of triples $(u,v,w) \in \Ft^* \times \Ft^* \times \Ft^*$ and $(x,y,z) \in \Ftm \times \Ftm \times \Ftm$, which satisfy the following two equations:
$$
\begin{cases}
  ux^{3^{\mmo}+1}+vy^{3^{\mmo}+1}+wz^{3^{\mmo}+1}=0, \\
  ux^{3^{\mmt}+1}+vy^{3^{\mmt}+1}+wz^{3^{\mmt}+1}=0.
\end{cases}
$$
It is easy to see that $\sum_{a,b \in \Ftm} S(a,b)^3=3^{2m}N_3$. Thus, it suffices to determine $N_3$.

If two of $x,y,z$ are zero, then the third one must also be zero. When $x=y=z=0$, we have eight choices of the triple $(u,v,w)$. When exactly one of $x,y,z$ equals $0$, we can use the result of ii). For instance, if $x=0$, then $u \in \{1,2\}$ and the system degenerates to
$$
\begin{cases}
  vy^{3^{\mmo}+1}+wz^{3^{\mmo}+1}=0, \\
  vy^{3^{\mmt}+1}+wz^{3^{\mmt}+1}=0.
\end{cases}
$$
By ii), we have $4(3^m-1)$ choices for the pairs $(v,w)$ and $(y,z)$. Thus, we have $8(3^m-1)$ choices for the triples $(u,v,w)$ and $(0,y,z)$. In total, when exactly one of $x,y,z$ equals $0$, we have $24(3^m-1)$ choices.

When $x,y,z$ are all nonzero, we have
\begin{eqnarray}
  u(\frac{x}{z})^{3^{\mmo}+1}+v(\frac{y}{z})^{3^{\mmo}+1}+w=0,  \label{eqn1} \\
  u(\frac{x}{z})^{3^{\mmt}+1}+v(\frac{y}{z})^{3^{\mmt}+1}+w=0.  \label{eqn2}
\end{eqnarray}
Raising to the $3^{\mpt}$-th power both sides of (\ref{eqn1}) and (\ref{eqn2}), we obtain
\begin{eqnarray}
  u(\frac{x}{z})^{3^{\mpt}+3}+v(\frac{y}{z})^{3^{\mpt}+3}+w=0,  \label{eqn11} \\
  u(\frac{x}{z})^{3^{\mpt}+1}+v(\frac{y}{z})^{3^{\mpt}+1}+w=0.  \label{eqn22}
\end{eqnarray}
Taking the difference between (\ref{eqn11}) and (\ref{eqn22}), we have
$$
u((\frac{x}{z})^{3^{\mpt}+3}-(\frac{x}{z})^{3^{\mpt}+1})+v((\frac{y}{z})^{3^{\mpt}+3}-(\frac{y}{z})^{3^{\mpt}+1})=0.
$$
Dividing $(\frac{y}{z})^{3^{\mpt}+1}$ on both sides of the above equation, we can easily obtain
\begin{equation}
(\frac{x}{y})^{3^{\mpt}+1}=-\frac{v(y^2-z^2)}{u(x^2-z^2)}. \label{eqn3}
\end{equation}
Raising to the third power  both sides of (\ref{eqn2}), we obtain
\begin{equation}
  u(\frac{x}{z})^{3^{\mmo}+3}+v(\frac{y}{z})^{3^{\mmo}+3}+w=0.  \label{eqn4}
\end{equation}
Taking the difference between (\ref{eqn4}) and (\ref{eqn1}), we have
$$
u((\frac{x}{z})^{3^{\mmo}+3}-(\frac{x}{z})^{3^{\mmo}+1})+v((\frac{y}{z})^{3^{\mmo}+3}-(\frac{y}{z})^{3^{\mmo}+1})=0
$$
Dividing $(\frac{y}{z})^{3^{\mmo}+1}$ on both sides of the above equation, we can easily obtain
\begin{equation}
(\frac{x}{y})^{3^{\mmo}+1}=-\frac{v(y^2-z^2)}{u(x^2-z^2)}. \label{eqn5}
\end{equation}
By (\ref{eqn3}) and (\ref{eqn5}), we have $(\frac{x}{y})^{3^{\mpt}+1}=(\frac{x}{y})^{3^{\mmo}+1}$, which is equivalent to
$(\frac{x}{y})^8=1$.
Since $\gcd(8,3^m-1)=2$, we have $(\frac{x}{y})^2=1$. By the symmetry of $x,y,z$, we also conclude that $(\frac{x}{z})^2=1$ and $(\frac{y}{z})^2=1$. Hence, the original system degenerates to
$$
u+v+w=0, \ \  u,v,w \in \gf(3)^*.
$$
It is easy to verify that there are two choices of the triple $(u,v,w)$ and $4(3^m-1)$ choices of the triple $(x,y,z)$. In total, there are $8(3^m-1)$ choices when $x,y,z$ are all nonzero. Hence, $N_3=8+24(3^m-1)+8(3^m-1)=32\times3^m-24$.

iv) The proof is very similar to that of ii), and thus omitted here.
\end{proof}

For $j\in\{0,2\}$, define
$$
n_j=\left|\left\{(a,b) \in \Ftm \times \Ftm : T(a,b)=\pm 3^{\frac{m+j}{2}}\sq\right\}\right|.
$$


For $j\in\{1,3\}$ and $\ep \in \{1,-1\}$, define
$$
n_{\ep,j}=\left|\left\{(a,b) \in \Ftm \times \Ftm : T(a,b)=\ep 3^{\frac{m+j}{2}}\right\}\right|.
$$

Under the assumption $n_{1,3}=0$, which will be verified later (see the proof of Theorem~\ref{thm-delta22} below),
we obtain the value distribution of $T(a,b)$ and $S(a,b)$.

\begin{lemma}\label{QF2}
Let $m \ge 3$ be an odd integer. Assume that $n_{1,3}=0$.
\begin{itemize}
\item[(1)] The value distribution of $T(a,b)$ is as follows: \\
\begin{center}
\ra{1.8}
\begin{tabular}{|c|c|c|}
\hline
\emph{Rank} $r_{a,b}$ & \emph{Value} $T(a,b)$ &  \emph{Multiplicity} \\ \hline
$m$ & $ 3^{\hm}\sq$     &  $\frac{(3^m-1)(8\times 3^m-9\times 3^{m-1}+9)}{8}$ \\ \hline
$m$ & $-3^{\hm}\sq$     &  $\frac{(3^m-1)(8\times 3^m-9\times 3^{m-1}+9)}{8}$ \\ \hline
$m-1$ & $3^{\mpo}$ & $\frac{(3^{m-1}+3^{\mmo})(3^m-1)}{2}$ \\ \hline
$m-1$ & $-3^{\mpo}$ & $\frac{(3^{m-1}-3^{\mmo})(3^m-1)}{2}$ \\ \hline
$m-2$ & $ 3^{\hm+1}\sq$ & $\frac{(3^m-1)(3^{m-1}-1)}{8}$ \\ \hline
$m-2$ & $-3^{\hm+1}\sq$ & $\frac{(3^m-1)(3^{m-1}-1)}{8}$ \\ \hline
$0$ & $3^{m}$ & $1$ \\ \hline
\end{tabular}
\end{center}
\vspace{5pt}
\item[(2)] The value distribution of $S(a,b)$ is as follows:\\
\begin{center}
\ra{1.8}
\begin{tabular}{|c|c|c|}
\hline
\emph{Rank} $r_{a,b}$ & \emph{Value} $S(a,b)$  &  \emph{Multiplicity} \\ \hline
$m, m-2$ & $0$     &  $(3^m-3^{m-1}+1)(3^m-1)$ \\ \hline
$m-1$ & $2\times3^{\mpo}$ & $\frac{(3^{m-1}+3^{\mmo})(3^m-1)}{2}$ \\ \hline
$m-1$ & $-2\times3^{\mpo}$ & $\frac{(3^{m-1}-3^{\mmo})(3^m-1)}{2}$ \\ \hline
$0$ & $2\times3^{m}$ & $1$ \\ \hline
\end{tabular}
\end{center}
\end{itemize}
\end{lemma}
\begin{proof}
By (\ref{eqn-ST}), we have
$$
S(a,b)=\begin{cases}
          0 & \mbox{if $r_{a,b} \in \{m,m-2\}$,} \\
          2T(a,b) & \mbox{if $r_{a,b} \in \{m-1,m-3\}$.}
       \end{cases}
$$
Therefore, the value distribution of $S(a,b)$ easily follows from that of $T(a,b)$. Note that $T(0,0)=3^m$ and $S(0,0)=2\times3^m$. Below, we are going to determine the value distribution of $T(a,b)$, employing the four moment identities in Lemma~\ref{moment}. Indeed, the moment identities lead to the following four equations:
\begin{align*}
\sum_{a,b\in\gf(3^m)}S(a,b)&=S(0,0)+2\sum_{r_{a,b} \in \{m-1,m-3\}}T(a,b) \\
&=2\times3^m+2(3^{\mpo}(n_{1,1}-n_{-1,1})+3^{\mpt}(n_{1,3}-n_{-1,3}))=2\times3^{2m}. \\
\sum_{a,b\in\gf(3^m)}S(a,b)^2&=S(0,0)^2+\sum_{r_{a,b} \in \{m-1,m-3\}}(2T(a,b))^2 \\ &=4\times3^{2m}+4(3^{m+1}(n_{1,1}+n_{-1,1})+3^{m+3}(n_{1,3}+n_{-1,3}))=4\times3^{3m}. \\
\sum_{a,b\in\gf(3^m)}S(a,b)^3&=S(0,0)^3+\sum_{r_{a,b} \in \{m-1,m-3\}}(2T(a,b))^3 \\ &=8\times3^{3m}+8(3^{\frac{3(m+1)}{2}}(n_{1,1}-n_{-1,1})+3^{\frac{3(m+3)}{2}}(n_{1,3}-n_{-1,3}))=32\times3^{3m}-24\times3^{2m}.\\ \sum_{a,b\in\gf(3^m)}T(a,b)^2&=T(0,0)^2+\sum_{(a,b)\ne(0,0)}T(a,b)^2 \\
&=3^{2m}-3^m n_0+3^{m+1}(n_{1,1}+n_{-1,1})-3^{m+2}n_2+3^{m+3}(n_{1,3}+n_{-1,3})=3^{2m}.
\end{align*}

Simplifying the above four equations leads to
\begin{align*}
n_{1,1}-n_{-1,1}+3(n_{1,3}-n_{-1,3})=3^{\mmo}(3^m-1), \\
n_{1,1}+n_{-1,1}+9(n_{1,3}+n_{-1,3})=3^{m-1}(3^m-1), \\
n_{1,1}-n_{-1,1}+27(n_{1,3}-n_{-1,3})=3^{\mmo}(3^m-1),\\
-n_0+3(n_{1,1}+n_{-1,1})-9n_2+27(n_{1,3}+n_{-1,3})=0.
\end{align*}
Moreover, by definition, we have
$$
n_0+n_2+n_{1,1}+n_{-1,1}+n_{1,3}+n_{-1,3}=3^{2m}-1.
$$
Together with $n_{1,3}=0$, we already have six linear equations with respect to $n_0, n_2, n_{1,1}, n_{-1,1}, n_{1,3}$ and $n_{-1,3}$, from which the value distribution of $T(a,b)$ easily follows.
\end{proof}

%
%
%

\begin{theorem}\label{thm-delta22}
Let $m \geq 3$. Then the ternary code $\tilde{\C}_{(n,3, m, \delta_2)}$ has parameters
$$
\left[ \frac{3^m-1}{2},\, k,\,   d  \right],
$$
where
\begin{eqnarray*}
k=\left\{ \begin{array}{ll}
2m & \mbox{ if $m$ is odd, } \\
\frac{3m}{2}  & \mbox{ if $m$ is even},
\end{array}
\right.
\end{eqnarray*}
and $d=3^{m-1}-3^{\lfloor \mmo \rfloor}$. Moreover, the weight distribution is presented in Table~\ref{table3} when $m$ is even and in Table~\ref{table4} when $m$ is odd.
\end{theorem}
\begin{proof}
The conclusion on the dimension of the code follows from Lemma \ref{lem-delta1}. $\tilde{\C}_{(n, 3, m, \delta_2)}$ has two nonzeroes and $\be^{\de_1}$, $\be^{\de_2}$ are two non-conjugate roots of its parity-check polynomial. Note that
\begin{align*}
\frac{3^m-1}{2}-\de_1 &= \frac{3^m-1}{2}-3^{m-1}+1+\frac{3^{\lfloor\frac{m-1}{2}\rfloor}-1}{2} \\ &=\frac{3^{m-1}+3^{\lfloor\frac{m-1}{2}\rfloor}}{2}, \\
\frac{3^m-1}{2}-\de_2 &= \frac{3^m-1}{2}-3^{m-1}+1+\frac{3^{\lfloor\frac{m+1}{2}\rfloor}-1}{2} \\ &=\frac{3^{m-1}+3^{\lfloor\frac{m+1}{2}\rfloor}}{2}.
\end{align*}
Define
\begin{align*}
c_1(a,b)&=\begin{cases}
\left(\Trtmt\left(a\be^{\frac{(3^{m-1}+3^{\frac{m-1}{2}})j}{2}}+b\be^{\frac{(3^{m-1}+3^{\frac{m+1}{2}})j}{2}}\right)\right)_{j=0}^{n-1} & \mbox{if $m \ge 3$ is odd} \\
\left(\Tr_3^{3^{\hm}}\left(a\be^{\frac{(3^{m-1}+3^{ \frac{m}{2}-1})j}{2}}\right)+\Trtmt\left(b\be^{\frac{(3^{m-1}+3^{\frac{m}{2}})j}{2}}\right)\right)_{j=0}^{n-1}
& \mbox{if $m \ge 4$ is even}
\end{cases}\\
&=\begin{cases}
\left(\Trtmt\left(a\al^{(3^{m-1}+3^{\frac{m-1}{2}})j}+b\al^{(3^{m-1}+3^{\frac{m+1}{2}})j}\right)\right)_{j=0}^{n-1} & \mbox{if $m \ge 3$ is odd} \\
\left(\Tr_3^{3^{\hm}}\left(a\al^{(3^{m-1}+3^{\frac{m}{2}-1})j}\right)+\Trtmt\left(b\al^{(3^{m-1}+3^{\frac{m}{2} })j}\right)\right)_{j=0}^{n-1} & \mbox{if $m \ge 4$ is even}
\end{cases}
\end{align*}
where $a,b \in \gf(3^m)$ for $m$ being odd and $a \in \gf(3^{\hm})$, $b \in \gf(3^m)$ for $m$ being even. By Proposition~\ref{prop-trace}, we have
$$
\tilde{\C}_{(n, 3, m, \delta_2)}=\{ c_1(a,b) : a,b \in \Ftm \}
$$
when $m$ is odd and
$$
\tilde{\C}_{(n, 3, m, \delta_2)}=\{ c_1(a,b) : a \in \Fthm, b \in \Ftm \}
$$
when $m$ is even.

When $m$ is even, we have
\begin{align*}
\Tr_3^{3^{\hm}}\left(a\al^{(3^{m-1}+3^{\frac{m}{2}-1})j}\right)&=\Tr_3^{3^{\hm}}\left(a^3\al^{(3^{\frac{m}{2}}+1)j}\right),\\
\Trtmt\left(b\al^{(3^{m-1}+3^{\frac{m}{2}})j}\right)&=\Trtmt\left(b^{3^{\hm}}\al^{(3^{\frac{m}{2}-1}+1)j}\right).
\end{align*}
Since $x\rightarrow x^3$ is a permutation of $\Fthm$ and $x\rightarrow x^{3^{\hm}}$ is a permutation of $\Ftm$, $\tilde{\C}_{(n, 3, m, \delta_2)}$ has the same weight distribution with the following code:
\begin{equation}\label{2:xi6}
\left\{\left(\Tr_3^{3^{\hm}}\left(a\al^{(3^{\hm}+1)j}\right)+\Trtmt\left(b\al^{(3^{\hm-1}+1)j}\right)\right)_{j=0}^{n-1}, \quad a \in \gf(3^{\frac{m}{2}}),b \in \gf(3^m)\right\}.
\end{equation}
We remark that the weight distribution of the following cyclic code has been studied in \cite[Theorem 2]{LTW}:
\begin{equation}\label{2:xi7}
\left\{\left(\Tr_3^{3^{\hm}}\left(a\al^{(3^{\hm}+1)j}\right)+\Trtmt\left(b\al^{(3^{\hm-1}+1)j}\right)\right)_{j=0}^{2n-1}, \quad a \in \gf(3^{\frac{m}{2}}),b \in \gf(3^m)\right\}.
\end{equation}
Note that there is a one-to-one correspondence between the codewords of (\ref{2:xi6}) and (\ref{2:xi7}). Indeed, given a codeword of (\ref{2:xi6}), the concatenation with its copy produces a codeword of (\ref{2:xi7}). Hence, the weight of a codeword in code (\ref{2:xi6}) is half of its corresponding codeword in code (\ref{2:xi7}). Therefore, the weight distribution of code (\ref{2:xi6}) easily follows from that of code (\ref{2:xi7}) \cite[Theorem 2]{LTW}. Thus, the weight distribution of $\tilde{\C}_{(n, 3, m, \delta_2)}$ is obtained and presented in Table~\ref{table3}. It is also known that the minimum distance is $d=3^{m-1}-3^{\hm-1}$.

When $m$ is odd, we need some extra work. Note that
\begin{align*}
\Trtmt\left(a\al^{(3^{m-1}+3^{\frac{m-1}{2}})j}\right)&=\Trtmt\left(a^{3^{\mpo}}\al^{(3^{\mmo}+1)j}\right),\\
\Trtmt\left(b\al^{(3^{m-1}+3^{\frac{m+1}{2}})j}\right)&=\Trtmt\left(b^{3^{\mmo}}\al^{(3^{\mmt}+1)j}\right).
\end{align*}
Since $x\rightarrow x^{3^{\mpo}}$ and $x\rightarrow x^{3^{\mmo}}$ are both permutations of $\Ftm$,
$\tilde{\C}_{(n, 3, m, \delta_2)}$ has the same weight distribution with the following code
$$
\left\{c_2(a,b) : a,b \in \Ftm \right\}.
$$
where
$$
c_2(a,b)=\left(\Trtmt\left(a\al^{(3^{\mmo}+1)j}+b\al^{(3^{\mmt}+1)j}\right)\right)_{j=0}^{n-1}.
$$
Note that
\begin{align}
w(c_2(a,b))&=n-|\{0 \le j \le n-1 \mid \Trtmt(a\al^{(3^{\mmo}+1)j}+b\al^{(3^{\mmt}+1)j})=0 \}| \notag\\
         &=n-\sum_{j=0}^{n-1}\frac13 \sum_{y \in \Ft} \z^{y(\Trtmt(a\al^{(3^{\mmo}+1)j}+b\al^{(3^{\mmt}+1)j}))}\notag\\
         &=n-\frac16\sum_{j=0}^{2n-1}\sum_{y \in \Ft} \z^{y(\Trtmt(a\al^{(3^{\mmo}+1)j}+b\al^{(3^{\mmt}+1)j}))}\notag\\
         &=n-\frac16\sum_{y \in \Ft}\sum_{x\in\Ftm^*} \z^{y(\Trtmt(ax^{3^{\mmo}+1}+bx^{3^{\mmt}+1}))}\notag\\
         &=n+\frac12-\frac16\sum_{y \in \Ft}\sum_{x\in\Ftm} \z^{\Trtmt(yax^{3^{\mmo}+1}+ybx^{3^{\mmt}+1})}\notag\\
         &=\frac{3^m-1}{2}+\frac12-\frac{3^m}{6}-\frac16\sum_{y \in \Ft^*}\sum_{x\in\Ftm} \z^{\Trtmt(yax^{3^{\mmo}+1}+ybx^{3^{\mmt}+1})}\notag\\
         &=3^{m-1}-\frac16\sum_{y \in \Ft^*} T(ya,yb)\notag\\
         &=3^{m-1}-\frac16 S(a,b) \label{eqn-wt}
\end{align}
Hence, we obtain the weight distribution of $\tilde{\C}_{(n, 3, m, \delta_2)}$ (see Table \ref{table4}) directly from the value distribution of $S(a,b)$ in Lemma \ref{QF2}, provided that $n_{1,3}=0$.

We finally prove that $n_{1,3}=0$. By (\ref{eqn-wt}) and Lemma~\ref{QF2}, $n_{1,3}$ is the frequency of codewords of $\tilde{\C}_{(n, q, m, \delta_2)}$ which have the weight $3^{m-1}-3^{\mpo}$. However, $\tilde{\C}_{(n, 3, m, \delta_2)}$ is a subcode of $\C_{(n, 3,m,\de_2)}$, whose minimum distance satisfies $d \ge \de_2=3^{m-1}-1-\frac{3^{\mpo}-1}{2}>3^{m-1}-3^{\mpo}$. Therefore, there must be no codeword of weight $3^{m-1}-3^{\mpo}$, which forces $n_{1,3}=0$. This concludes the proof of Theorem \ref{thm-delta22}.
\end{proof}

\begin{table}
\ra{1.9}
\begin{center}
\caption{The weight distribution of $\tilde{\C}_{(n,3,m,\de_2)}$ when $m\ge4$ is even}
\begin{tabular}{|c|c|}
\hline
Weight  &  Frequency \\ \hline
$0$     &  $1$ \\ \hline
$3^{m-1}-3^{\hm-1}$ & $\frac{3(3^{\hm}-1)(3^{\hm}+1)^2}{8}$ \\ \hline
$3^{m-1}$ & $3^{\hm-1}(3^m-1)$ \\ \hline
$3^{m-1}+3^{\hm-1}$ & $\frac{3(3^{\hm}-1)(3^{m-1}+1)}{4}$ \\ \hline
$3^{m-1}+3^{\hm}$ & $\frac{(3^{\hm-1}-1)(3^m-1)}{8}$ \\ \hline
\end{tabular}
\label{table3}
\end{center}
\end{table}

\begin{table}
\ra{1.9}
\begin{center}
\caption{The weight distribution of $\tilde{\C}_{(n,3,m,\de_2)}$ when $m\ge3$ is odd}
\begin{tabular}{|c|c|}
\hline
Weight  &  Frequency \\ \hline
$0$     &  $1$ \\ \hline
$3^{m-1}-3^{\mmo}$ & $\frac{(3^{m-1}+3^{\mmo})(3^m-1)}{2}$ \\ \hline
$3^{m-1}$ & $(3^m-3^{m-1}+1)(3^m-1)$ \\ \hline
$3^{m-1}+3^{\mmo}$ & $\frac{(3^{m-1}-3^{\mmo})(3^m-1)}{2}$ \\ \hline
\end{tabular}
\label{table4}
\end{center}
\end{table}

\begin{example}
Let $(q, m)=(3,4)$. Then the code $\tilde{\C}_{(n, q, m, \delta_2)}$ of Theorem \ref{thm-delta22} has parameters
$[40, 6, 24]$, and weight enumerator
$
1+300z^{24} + 240z^{27} + 168z^{30} + 20z^{36}.
$
This is the best cyclic code and optimal according to {\rm \cite[p. 305]{Dingbk15}}.
\end{example}

\begin{example}
Let $(q, m)=(3,5)$. Then the code $\tilde{\C}_{(n, q, m, \delta_2)}$ of Theorem \ref{thm-delta22} has parameters
$[121, 10, 72]$, and weight enumerator
$
1+10890z^{72} +39446z^{81}+8712z^{90}.
$
This code has the same parameters as the best ternary linear code known in the Database.
\end{example}

\begin{theorem}\label{thm-delta2}
Let $m \geq 3$. Then the ternary code  $\C_{(n, 3, m, \delta_2)}$ has parameters
$$
\left[ \frac{3^m-1}{2},\, k,\, \de_2 \right],
$$
where
\begin{eqnarray*}
k=\left\{ \begin{array}{ll}
2m+1 & \mbox{ if $m$ is odd, } \\
\frac{3m+2}{2}  & \mbox{ if $m$ is even.}
\end{array}
\right.
\end{eqnarray*}
In addition, the weight distribution is presented in Table~\ref{table5} when $m$ is even and in Table~\ref{table6} when $m$ is odd.
\end{theorem}
\begin{proof}
The conclusion on the dimension of the code follows from Lemma \ref{lem-delta1}. Note that $\C_{(n, 3, m, \delta_2)}$ has three nonzeroes and $1$, $\be^{\de_1}$, $\be^{\de_2}$ are three roots of its parity-check polynomial, such that every two of them are non-conjugate. Define
\begin{align*}
c_1(a,b,c)&=\begin{cases}
\left(\Trtmt\left(a\be^{\frac{(3^{m-1}+3^{\frac{m-1}{2}})j}{2}}+b\be^{\frac{(3^{m-1}+3^{\frac{m+1}{2}})j}{2}}\right)+c\right)_{j=0}^{n-1} & \mbox{if $m \ge 3$ is odd} \\
\left(\Tr_3^{3^{\hm}}\left(a\be^{\frac{(3^{m-1}+3^{ \frac{m}{2}-1})j}{2}}\right)+\Trtmt\left(b\be^{\frac{(3^{m-1}+3^{\frac{m}{2}})j}{2}}\right)+c\right)_{j=0}^{n-1}
& \mbox{if $m \ge 4$ is even}
\end{cases}\\
&=\begin{cases}
\left(\Trtmt\left(a\al^{(3^{m-1}+3^{\frac{m-1}{2}})j}+b\al^{(3^{m-1}+3^{\frac{m+1}{2}})j}\right)+c\right)_{j=0}^{n-1} & \mbox{if $m \ge 3$ is odd} \\
\left(\Tr_3^{3^{\hm}}\left(a\al^{(3^{m-1}+3^{\frac{m}{2}-1})j}\right)+\Trtmt\left(b\al^{(3^{m-1}+3^{\frac{m}{2} })j}\right)+c\right)_{j=0}^{n-1} & \mbox{if $m \ge 4$ is even}
\end{cases}
\end{align*}
where $a,b \in \gf(3^m)$, $c\in\Ft$ for $m$ being odd and $a \in \gf(3^{\hm})$, $b \in \gf(3^m)$, $c\in\Ft$ for $m$ being even. Similar to the proof of Theorem~\ref{thm-delta22}, by using Proposition~\ref{prop-trace}, we have
$$
\C_{(n, 3, m, \delta_2)}=\{ c_1(a,b,c) : a,b \in \Ftm, c\in\Ft \}
$$
when $m$ is odd and
$$
\C_{(n, 3, m, \delta_2)}=\{ c_1(a,b,c) : a \in \Fthm, b \in \Ftm, c\in\Ft \}
$$
when $m$ is even.

When $m$ is even, with the same argument as in the proof of Theorem~\ref{thm-delta22}, we know that $\C_{(n, q, m, \delta_2)}$ has the same weight distribution with the following code
$$
\left\{c_2(a,b,c) : a \in \Fthm, b \in \Ftm, c \in \Ft\right\},
$$
where
$$
c_2(a,b,c)=\left(\Tr_3^{3^{\hm}}\left(a\al^{(3^{\hm}+1)j}\right)+\Trtmt\left(b\al^{(3^{\hm-1}+1)j}\right)+c\right)_{j=0}^{n-1}.
$$

Recall that $\eta_0$ is the quadratic character of $\Ft$. With the help of Lemma~\ref{lemmaquad}, we have
\begin{align}
w(c_2(a,b,c))&=n-|\{0 \le j \le n-1 \mid \Tr_3^{3^{\hm}}(a\al^{(3^{\hm}+1)j})+\Trtmt(b\al^{(3^{\hm-1}+1)j})+c=0 \}| \notag\\
         &=n-\sum_{j=0}^{n-1}\frac13 \sum_{y \in \Ft} \z^{y(\Tr_3^{3^{\hm}}(a\al^{(3^{\hm}+1)j})+\Trtmt(b\al^{(3^{\hm-1}+1)j})+c)}\notag\\
         &=n-\frac16\sum_{j=0}^{2n-1}\sum_{y \in \Ft} \z^{y(\Tr_3^{3^{\hm}}(a\al^{(3^{\hm}+1)j})+\Trtmt(b\al^{(3^{\hm-1}+1)j})+c)}\notag\\
         &=n-\frac16\sum_{y \in \Ft}\sum_{x\in\Ftm^*} \z^{y(\Tr_3^{3^{\hm}}(ax^{3^{\hm}+1})+\Trtmt(bx^{3^{\hm-1}+1})+c)}\notag\\
         &=n+\frac16\sum_{y \in \Ft}\z^{yc}-\frac16\sum_{y \in \Ft}\sum_{x\in\Ftm} \z^{\Tr_3^{3^{\hm}}(yax^{3^{\hm}+1})+\Trtmt(ybx^{3^{\hm-1}+1})+yc}\notag\\
         &=\frac{3^m-1}{2}+\frac12\de_{0,c}-\frac{3^m}{6}-\frac16\sum_{y \in \Ft^*}\sum_{x\in\Ftm} \z^{\Tr_3^{3^{\hm}}(yax^{3^{\hm}+1})+\Trtmt(ybx^{3^{\hm-1}+1})+yc}\notag\\
         &=3^{m-1}+\frac12(\de_{0,c}-1)-\frac16\sum_{y \in \Ft^*}\z^{yc}\sum_{x\in\Ftm} \z^{\Tr_3^{3^{\hm}}(yax^{3^{\hm}+1})+\Trtmt(ybx^{3^{\hm-1}+1})}\notag\\
         &=3^{m-1}+\frac12 \left(\de_{0,c}-1\right)-\frac16 \sum_{y \in \Ft^*} \z^{yc} U(ya,yb)\notag\\
         &=3^{m-1}+\frac12 \left(\de_{0,c}-1\right)-\frac16 \sum_{y \in \Ft^*} \z^{yc} \eta_0(y)^{r_{a,b}'} U(a,b) \notag
\end{align}
where
\begin{align*}
U(a,b)&=\sum_{x \in \Ftm} \z^{\Trtmt(ax^{3^{\hm}+1}+bx^{3^{\hm-1}+1})}, \\
\de_{0,c}&=\begin{cases}
  1 & \mbox{if $c=0$,} \\
  0 & \mbox{if $c\ne0$,}
\end{cases}
\end{align*}
and $r_{a,b}'$ is the rank of the quadratic form $\Trtmt(ax^{3^{\hm}+1}+bx^{3^{\hm-1}+1})$. Thanks to \cite[Theorem 1]{LTW}, the value distribution of $U(a,b)$ is already known which is presented in the following table:
\begin{center}
\ra{1.75}
\begin{tabular}{|c|c|c|}
\hline
Rank $r_{a,b}'$ & Value $U(a,b)$ &  Multiplicity \\ \hline
$m$ & $3^{\hm}$     &  $\frac{3(3^{\hm}-1)(3^{\hm}+1)^2}{8}$ \\ \hline
$m$ & $-3^{\hm}$ & $\frac{3(3^{\hm}-1)(3^{m-1}+1)}{4}$ \\ \hline
$m-1$ & $ 3^{\frac{m+1}{2}}i$ & $\frac{3^{\hm-1}(3^m-1)}{2}$ \\ \hline
$m-1$ & $-3^{\frac{m+1}{2}}i$ & $\frac{3^{\hm-1}(3^m-1)}{2}$ \\ \hline
$m-2$ & $-3^{\hm+1}$ & $\frac{(3^{\hm-1}-1)(3^{m}-1)}{8}$ \\ \hline
$0$ & $3^{m}$ & $1$ \\ \hline
\end{tabular}
\end{center}
It is easy to see that the weight distribution of $\C_{(n, q, m, \delta_2)}$ (see Table \ref{table5}) follows directly from the value distribution of $U(a,b)$ when $m$ is even. For example, according to the table above, when $r_{a,b}'=m$, $U(a,b)$ takes the value $3^{\hm}$ for $\frac{3(3^{\hm}-1)(3^{\hm}+1)^2}{8}$ times. Thus, if $c=0$,  $w(c_2(a,b,c))$ takes the value $3^{m-1}-3^{\hm-1}$ for $\frac{3(3^{\hm}-1)(3^{\hm}+1)^2}{8}$ times, and if $c=1$ or $2$,  $w(c_2(a,b,c))$ takes the value $3^{m-1}+\frac12 \left(3^{\hm-1}-1\right)$ for $\frac{3(3^{\hm}-1)(3^{\hm}+1)^2}{4}$ times.

When $m$ is odd, with the same argument as in the proof of Theorem~\ref{thm-delta22}, we know that $\C_{(n, q, m, \delta_2)}$ shares the same weight distribution with the following code
$$
\left\{c_3(a,b,c) : a,b \in \Ftm, c \in \Ft\right\},
$$
where
$$
c_3(a,b,c)=\left(\Trtmt\left(a\al^{(3^{\mmo}+1)j}+b\al^{(3^{\mmt}+1)j}\right)+c\right)_{j=0}^{n-1}.
$$
Very similar to the above case of $w(c_2(a,b,c))$, we can show that
$$
w(c_3(a,b,c))=3^{m-1}+\frac12 \left(\de_{0,c}-1\right)-\frac16 \sum_{y \in \Ft^*} \z^{yc}\eta_0(y)^{r_{a,b}} T(a,b).
$$
Employing the value distribution of $T(a,b)$ in Lemma~\ref{QF2}, we obtain the weight distribution of $\C_{(n, 3, m, \delta_2)}$ (see Table~\ref{table6}) directly when $m$ is odd. This completes the proof of Theorem \ref{thm-delta2}.
\end{proof}

\begin{table}
\ra{1.8}
\begin{center}
\caption{The weight distribution of $\C_{(n,3,m,\de_2)}$ when $m \ge 4$ is even}
\begin{tabular}{|c|c|}
\hline
Weight  &  Frequency \\ \hline
$0$     &  $1$ \\ \hline
$\frac{3^m-3^{m-1}-3^{\hm}-1}{2}$ & $\frac{(5\times3^{\hm-1}-1)(3^m-1)}{4}$ \\ \hline
$\frac{3^m-3^{m-1}-2\times3^{\hm-1}}{2}$ & $\frac{3(3^{\hm}-1)(3^{\hm}+1)^2}{8}$ \\ \hline
$\frac{3^m-3^{m-1}-3^{\hm-1}-1}{2}$ & $\frac{3(3^{\hm}-1)(3^{m-1}+1)}{2}$ \\ \hline
$\frac{3^m-3^{m-1}}{2}$ & $3^{\hm-1}(3^m-1)$ \\ \hline
$\frac{3^m-3^{m-1}+3^{\hm-1}-1}{2}$ & $\frac{3(3^{\hm}-1)(3^{\hm}+1)^2}{4}$ \\ \hline
$\frac{3^m-3^{m-1}+2\times3^{\hm-1}}{2}$ & $\frac{3(3^{\hm}-1)(3^{m-1}+1)}{4}$ \\ \hline
$\frac{3^m-3^{m-1}+3^{\hm}-1}{2}$ & $3^{\hm-1}(3^m-1)$ \\ \hline
$\frac{3^m-3^{m-1}+2\times3^{\hm}}{2}$ & $\frac{(3^{\hm-1}-1)(3^m-1)}{8}$ \\ \hline
$\frac{3^m-1}{2}$ & $2$ \\ \hline
\end{tabular}
\label{table5}
\end{center}
\end{table}

\begin{table}
\ra{1.8}
\begin{center}
\caption{The weight distribution of $\C_{(n, 3, m, \de_2)}$ when $m\ge3$ is odd}
\begin{tabular}{|c|c|}
\hline
Weight  &  Frequency \\ \hline
$0$     &  $1$ \\ \hline
$\frac{3^m-3^{m-1}-3^{\mpo}-1}{2}$ & $\frac{(3^{m-1}-1)(3^m-1)}{8}$ \\ \hline
$\frac{3^m-3^{m-1}-2\times3^{\mmo}}{2}$ & $\frac{(3^{m-1}+3^{\mmo})(3^{m}-1)}{2}$ \\ \hline
$\frac{3^m-3^{m-1}-3^{\mmo}-1}{2}$ & $\frac{(8\times3^m-3^{m-1}-8\times3^{\mmo}+9)(3^{m}-1)}{8}$ \\ \hline
$\frac{3^m-3^{m-1}}{2}$ & $(3^{m}-3^{m-1}+1)(3^m-1)$ \\ \hline
$\frac{3^m-3^{m-1}+3^{\mmo}-1}{2}$ & $\frac{(8\times3^m-3^{m-1}+8\times3^{\mmo}+9)(3^{m}-1)}{8}$ \\ \hline
$\frac{3^m-3^{m-1}+2\times3^{\mmo}}{2}$ & $\frac{(3^{m-1}-3^{\mmo})(3^{m}-1)}{2}$ \\ \hline
$\frac{3^m-3^{m-1}+3^{\mpo}-1}{2}$ & $\frac{(3^{m-1}-1)(3^m-1)}{8}$ \\ \hline
$\frac{3^m-1}{2}$ & $2$ \\ \hline
\end{tabular}
\label{table6}
\end{center}
\end{table}

\begin{example}
Let $(q, m)=(3,4)$. Then the code $\C_{(n, q, m, \delta_2)}$ of Theorem \ref{thm-delta2} has parameters
$[40, 7, 22]$, and weight enumerator
$$
1+280z^{22} + 300z^{24} + 336z^{25} + 240z^{27} + 600 z^{28} + 168z^{30} + 240z^{31} + 20z^{36} + 2z^{40}.
$$
According to {\rm \cite[p. 305]{Dingbk15}}, this is the best ternary cyclic code, and has the same parameters as the best ternary
linear code in the Database. Note that for each ternary linear code with parameters $[40, 7, d]$, we have $d \leq 23$.
\end{example}

\begin{example}
Let $(q, m)=(3,5)$. Then the code $\C_{(n, q, m, \delta_2)}$ of Theorem \ref{thm-delta2} has parameters
$[121, 11, 67]$, and weight enumerator
$$
1+2420z^{67}+10890z^{72}+54450z^{76}+39446z^{81}+58806z^{85}+8712z^{90}+2420z^{94}+2z^{121}.
$$
The best ternary linear code known in the Database has parameters $[121, 11, 68]$, which is not known to be cyclic.
\end{example}

\section{Some narrow-sense BCH codes with special designed distances}\label{sec-otherforms}

In this section, for general prime power $q$, we focus on narrow-sense BCH codes of length $n=\frac{q^m-1}{q-1}$ with special designed distances. As stated in Section~\ref{sec-pre}, some information about coset leaders modulo $\frac{q^m-1}{q-1}$ can be obtained via the NDS decomposition. Employing Proposition~\ref{prop-Bose}, we can obtain the Bose distance if the designed distance is of certain special form.

\begin{theorem}\label{thm-spedes}
Let $2 \le \de \le n$ be an integer. Suppose $E(\de)=\ul{V_1}\ul{V_2}\ldots\ul{V_r}$.
\begin{itemize}
\item[i)] Suppose $\ol{\de}$ has only $0$ and $1$ as components. If $\ul{V_1}=\ul{V_2}=\cdots=\ul{V_r}$ or $\ul{V_1}=\ul{V_2}=\cdots=\ul{V_j}<\ul{V_k}$ for some $j$ and all $k$ such that $j < k \le r$, then $\C_{(n,q,m,\de)}$ has Bose distance $d_B=\de$. In particular, if $E(\de)=\ul{V_1}$, then $\C_{(n,q,m,\de)}$ has Bose distance $d_B=\de$.
\item[ii)] Suppose $\ul{V_1}$ has length $l$ and has components either $0$ or $1$. Let $\ul{V_1}>\ul{V_2}$ and $m=al+b$, where $0 \le b \le l-1$.
    \begin{itemize}
    \item[(1)] If $b=0$, then $\C_{(n,q,m,\de)}$ has Bose distance
$$
d_B = E^{-1}(\underbrace{\ul{V_1}\ul{V_1}\ldots\ul{V_1}}_{a}).
$$
\item[(2)]
If $1 \le b \le l-1$, then $\C_{(n,q,m,\de)}$ has Bose distance
$$
d_B \ge E^{-1}(\underbrace{\ul{V_1}\ul{V_1}\ldots\ul{V_1}}_{a}S(T_b(\ul{V_1}))).
$$
In particular, if the last component of $S(T_b(\ul{V_1}))$ is $1$, then the equality holds.
\end{itemize}
\end{itemize}
\end{theorem}
\begin{proof}
These results are immediate consequences of Proposition~\ref{prop-Bose}, Lemma~\ref{lem-NDS1} and Lemma~\ref{lem-NDS2}.
\end{proof}

The theorem above is very powerful in determining the Bose distance of $\C_{(n,q,m,\de)}$ for some special $\de$. In the following, we give several examples.

\begin{example}
For $q=3$, $m=6$ and $\de=110$, consider the code $\C_{(364, 3,6,110)}$. Note that
$$
\ol{\de}=(0,1,1,0,0,2)=\ul{V_1}\ul{V_2},
$$
where $\ul{V_1}=(0,1,1)$ and $\ul{V_2}=(0,0,2)$. Since $\ul{V_1}>\ul{V_2}$, by ii) of Theorem~\ref{thm-spedes},
the Bose distance
$$
d_B = E^{-1}(\ul{V_1}\ul{V_1})=E^{-1}(0,1,1,0,1,1)=112.
$$
Indeed, the smallest coset leader no less than $\de=110$ is just $112$, where $C_{112}=\{112,280,336\}$.
\end{example}

\begin{example}
For $q=3$, $m=5$ and $\de=29$, consider the code $\C_{(121,3,5,29)}$. Note that
$$
\ol{\de}=(0,1,0,0,2)=\ul{V_1}\ul{V_2},
$$
where $\ul{V_1}=(0,1)$ and $\ul{V_2}=(0,0,2)$. Since $\ul{V_1}>\ul{V_2}$, by ii) of Theorem~\ref{thm-spedes},
the Bose distance
$$
d_B = E^{-1}(\ul{V_1}\ul{V_1}S(T_1(\ul{V_1})))=E^{-1}(0,1,0,1,1)=31.
$$
Indeed, the smallest coset leader no less than $\de=29$ is just $31$, where $C_{31}=\{31,37,91,93,111\}$.
\end{example}

\begin{example}
For $q=7$, $m=5$ and $\de=393$, consider the code $\C_{(2801,7,5,393)}$. Note that
$$
\ol{\de}=(0,1,1,0,1)=\ul{V_1}\ul{V_2},
$$
where $\ul{V_1}=(0,1,1)$ and $\ul{V_2}=(0,1)$. Since $\ul{V_1}>\ul{V_2}$, by ii) of Theorem~\ref{thm-spedes},
the Bose distance
$$
d_B \ge E^{-1}(\ul{V_1}S(T_2(\ul{V_1})))=E^{-1}(0,1,1,0,2)=394.
$$
Indeed, the smallest coset leader no less than $\de=393$ is just $394$, where $C_{394}=\{394,694,2057,2500,2758\}$. Hence, we have the Bose distance $d_B=394$.
\end{example}

Below, we consider two special classes of designed distances. First, we consider BCH codes with designed distances $\frac{q^i-1}{q-1}$, where $1 \le i \le m-1$.

\begin{theorem}\label{thm-spedesdis1}
For $1 \le i \le m-1$, $\C_{(n,q,m,\frac{q^i-1}{q-1})}$ has Bose distance $d_B=\frac{q^i-1}{q-1}$. Furthermore, if $1 \le i \le \lceil \frac{m}{2} \rceil$, then the code $\C_{(n,q,m,\frac{q^i-1}{q-1})}$ has parameters
$$
\left[\frac{q^m-1}{q-1}, \, \frac{q^m-1}{q-1}-m(q^{i-1}-1), \, d\right],
$$
where $d \ge \frac{q^i-1}{q-1}$. In particular, if $i \mid m$, then $d=\frac{q^i-1}{q-1}$.
\end{theorem}
\begin{proof}
Note that the sequence
$$
\ol{\frac{q^i-1}{q-1}}=(\underbrace{0,\ldots,0}_{m-i},1,\ldots,1)
$$
is an NDS. By iii) of Lemma~\ref{lem-NDS1}, $\frac{q^i-1}{q-1}$ is a coset leader. Then, by Proposition~\ref{prop-Bose}, we have $d_B=\frac{q^i-1}{q-1}$. When $1 \le i \le \lceil \frac{m}{2} \rceil$, the conclusion on dimensions follows from Theorem~\ref{thm-dim}. If $i \mid m$, by Lemma~\ref{lem-mineqdes}, we have $d=\frac{q^i-1}{q-1}$.
\end{proof}

It would be good if the following open problem can be settled.

\begin{problem}
Dose the code $\C_{(n, q, m, (q^i-1)/(q-1))}$ have minimum distance $d=(q^i-1)/(q-1)$, where $1 \leq i \leq m-1$ ?
\end{problem}

Note that Theorem~\ref{thm-spedesdis1} gives an affirmative answer to the open problem, when $i \mid m$.
Next, we consider BCH codes with designed distances $q^i+l$, where $1 \le i \le \lceil \frac{m}{2} \rceil-1$ and $1 \le l \le q-1$.

\begin{theorem}\label{thm-spedesdis2}
For $1 \le i \le \lceil \frac{m}{2} \rceil-1$ and $1 \le l \le q-1$, $\C_{(n,q,m,q^i+l)}$ has Bose distance $d_B=q^i+l$. Furthermore, the code $\C_{(n,q,m,q^i+l)}$ has parameters
$$
\left[\frac{q^m-1}{q-1}, \, \frac{q^m-1}{q-1}-m\left\lceil (q^i+l-1)\left(1-\frac{1}{q}\right) \right\rceil, \, d\right],
$$
where $d \ge q^i+l$. In particular, if $(q^i+l) \mid n$, then $d=q^i+l$.
\end{theorem}
\begin{proof}
The conclusion on the dimension follows from Theorem~\ref{thm-dim}. For $1 \le i \le \lceil \frac{m}{2} \rceil-1$ and $1 \le l \le q-1$, let $\de=q^i+l$. To prove that the Bose distance is equal to $\de$, by Proposition~\ref{prop-Bose}, it suffices to show that $\de$ is a coset leader. Note that
$$
\ol{\de}=(\underbrace{0,\ldots,0}_{m-i-1},1,\underbrace{0,\ldots,0}_{i-1},l).
$$
We are going to show that $\de$ is a coset leader by analyzing $\ol{\de}$. Direct computation shows that for $1 \le i \le m-2$,
$\ol{q^i\de}>\ol{\de}$. Moreover,
$$
\ol{q^{m-1}\de}=\begin{cases}
    (1,\underbrace{0,\ldots,0}_{m-i-1},1,\underbrace{0,\ldots,0}_{i-1}) & \mbox{if $l=1$},  \\
    (0,\underbrace{q-l,\ldots,q-l}_{m-i-1},q-l+1,\underbrace{q-l,\ldots,q-l}_{i-2},q-l+1) & \mbox{if $2 \le l \le q-1$,}
    \end{cases}
$$
which implies that $\ol{q^{m-1}\de} > \ol{\de}$. Consequently, $\de$ is a coset leader modulo $n$. In addition, if $(q^i+l) \mid n$, by Lemma~\ref{lem-mineqdes}, we have $d=q^i+l$.
\end{proof}

\begin{example}
Let $(q, m)=(3,3)$. Then the code $\C_{(n,q, m, q+1)}$ has parameters
$[13, 7, 4]$. The optimal linear code in the Database has parameters $[13,7,5]$, which is not known to be cyclic.
\end{example}

\begin{example}
Let $(q, m)=(3,4)$. Then the code $\C_{(n,q, m, q+1)}$ has parameters
$[40, 32, 4]$, and is the best ternary cyclic code according to {\rm \cite[p. 306]{Dingbk15}}.
The optimal linear code in the Database has parameters $[40,32,5]$, which is not known to be cyclic.
\end{example}

\begin{example}
Let $(q, m)=(3,4)$. Then the code $\C_{(n, q, m, q+2)}$ has parameters
$[40, 28, 5]$. The best linear code in the Database has parameters $[40,28,6]$, which is not known to be cyclic.
\end{example}

\begin{example}
Let $(q, m)=(3,5)$. Then the code $\C_{(n, q, m, q+2)}$ has parameters
$[121, 106, 6]$. The best linear code in the Database has the same parameters, and is not known to be cyclic.
\end{example}

It would be nice if the following open problem on the minimum distance of $\C_{(n, q, m, q+1)}$ could
be settled.

\begin{problem}
Does the code $\C_{(n, q, m, q+1)}$ have minimum distance $d=q+1$?
\end{problem}

Note that Theorem~\ref{thm-spedesdis2} provides an affirmative answer to the open problem in the case that $m>2$ is even.

\section{Open problems and concluding remarks}

Although BCH codes are introduced in almost every book on coding theory, a very small number of results about them are
available in the literature (see \cite{ACS92, AS94, Charp90, Charp98, DDZ15} for information). In general, it is a hard problem
to determine the dimension of a BCH code, and it is much harder to find its minimum distance.

The known results on BCH codes are almost entirely for the primitive length $n=q^m-1$. To our knowledge, there are only a few papers on BCH codes with non-primitive lengths in the literature. This is because it is harder to deal with BCH codes with non-primitive lengths. This paper initializes the study of narrow-sense BCH codes of length $n=(q^m-1)/(q-1)$, and has the following contributions:
\begin{itemize}
\item The parameters of some narrow-sense BCH codes with large dimensions were determined in Section
         \ref{sec-largedim}.
\item The parameters of some narrow-sense BCH ternary codes with small dimensions were settled in Section
         \ref{sec-smalldim}. Specifically, we determined the weight distributions of the ternary BCH codes
         $\C_{(n,q,m,\delta_1)}$, $\tilde{\C}_{(n, q, m, \delta_1)}$, $\C_{(n,q,m,\delta_2)}$ and
         $\tilde{\C}_{(n, q, m, \delta_2)}$ in Theorems \ref{thm-delta1}, \ref{thm-delta11}, \ref{thm-delta2} and
         \ref{thm-delta22}, respectively.  A class of optimal BCH ternary codes were identified.
\item The parameters of some narrow-sense BCH codes with designed distances of special forms were settled
          in Section \ref{sec-otherforms}.
\end{itemize}

This paper only initialized the investigation of narrow-sense BCH codes of length $\frac{q^m-1}{q-1}$ over finite fields. There are many open problems on these codes. Below we mention a few open problems regarding these codes.

\begin{problem}
For $q=3$, determine the parameters of
$\C_{(n,q,m,\delta_i)}$ and $\tilde{\C}_{(n, q, m, \delta_i)}$ for $3 \leq i \leq \lceil m/4 \rceil$.
\end{problem}

For the case $q=3$, we did find $\delta_i$ for all $i$ with
$3 \leq i \leq \lceil m/4 \rceil$ in Section \ref{sec-otherdeltai}.
The dimensions of $\C_{(n,3,m,\delta_i)}$ and $\tilde{\C}_{(n, 3, m, \delta_i)}$ can thus
be determined for $i$ with $3 \leq i \leq \lceil m/4 \rceil$ recursively, given the dimensions of
$\C_{(n,3,m,\delta_2)}$ and $\tilde{\C}_{(n, 3, m, \delta_2)}$ computed earlier in this paper.
Specifically, for $i \geq 1$ we have
$$
\dim(\C_{(n,3,m,\delta_{i+1})})=\dim(\C_{(n,3,m,\delta_{i})})+|C_{\delta_i}| \mbox{ and }
\dim(\tilde{\C}_{(n,3,m,\delta_{i+1})})=\dim(\tilde{\C}_{(n,3,m,\delta_{i})})+|C_{\delta_i}|.
$$
The remaining task is to compute the cardinalities of the cyclotomic cosets $C_{\delta_i}$.
But it would be very hard to determine the
minimum distances of the ternary BCH codes $\C_{(n,3,m,\delta_i)}$ and $\tilde{\C}_{(n, 3, m, \delta_i)}$ for $i\ge 3$.

The case that $q \geq 4$ is much more complicated. We have the following open problems for this case.

\begin{problem}
For $q>3$, find the largest coset leader $\delta_1$ and determine the parameters of
$\C_{(n,q,m,\delta_1)}$ and $\tilde{\C}_{(n, q, m, \delta_1)}$.
\end{problem}


\begin{problem}
For $q>3$, find the second largest coset leader $\delta_2$ and determine the parameters of
$\C_{(n,q,m,\delta_2)}$ and $\tilde{\C}_{(n, q, m, \delta_2)}$.
\end{problem}


The codes $\C_{(n, q, m, (q^i-1)/(q-1))}$ have
also very good parameters according to our Magma examples. Hence, it is worthy to attack the following open problem.

\begin{problem}
Determine the dimension of $\C_{(n, q, m, (q^i-1)/(q-1))}$ for
$\lceil \frac{m}{2} \rceil < i  \leq  m-1$.
\end{problem}

It is possible to find the dimension of the code $\C_{(n, q, m, (q^{m-1}-1)/(q-1))}$. Examples of this code
suggest that the dimension of this code is lower bounded by
$$
\binom{p+m-2}{m-1}^s+1,
$$
where $q=p^s$ and $p$ is a prime. In general, the dimension of this code is much larger than the lower bound above. The reader is very welcome to settle the open problems in this paper.

While primitive narrow-sense BCH codes contain many good linear codes \cite{Dingbk15,DDZ15}, as shown by many examples in this paper, narrow-sense BCH codes of length $n=(q^m-1)/(q-1)$ also include many optimal linear codes. Very recently, new infinite families of $2$-designs and $3$-designs from linear codes are presented in \cite{D,DL}. The ternary narrow-sense BCH codes studied in Section~\ref{sec-smalldim} can be employed for  constructing $2$-designs and
some Steiner systems \cite{DL}. These nice applications are some of the motivations for studying narrow-sense BCH
codes of length $n=(q^m-1)/(q-1)$.

Finally, we point out an application of some of the ternary codes of this paper in secret sharing. Every linear code over
$\gf(q)$ can be employed to construct secret sharing schemes \cite{ADHK,CDY05,Mass93,YD06}. In order to make such
secret sharing scheme to have interesting access structures, we need a linear code $\C$ over $\gf(q)$ such that
\begin{eqnarray}\label{eqn-sss}
\frac{w_{min}}{w_{max}}>\frac{q-1}{q},
\end{eqnarray}
where $w_{max}$ and $w_{min}$ denote the maximum and minimum nonzero weight in $\C$, respectively.

The ternary codes of Theorems \ref{thm-delta1} and \ref{thm-delta22} satisfy the inequality in (\ref{eqn-sss})
when $m \geq 5$, and the codes of Theorem \ref{thm-delta11} have only one nonzero weight and obviously
satisfy the inequality in (\ref{eqn-sss}). Therefore, all the codes in Theorems \ref{thm-delta1},  \ref{thm-delta11},
and \ref{thm-delta22} can be employed to obtain secret sharing schemes with interesting access structures using
the framework documented in \cite{ADHK,CDY05,Mass93,YD06}.

\section*{Acknowledgements}

The authors are very grateful to the reviewers and the Associate Editor, Dr. Jyrki Lahtonen, for their detailed comments and suggestions that much improved the presentation and quality of this paper. The second author thanks Dr. Pascale Charpin for providing him with helpful information on narrow-sense primitive BCH codes.

\end{document}